\documentclass[11pt, a4paper]{article}

\usepackage[utf8]{inputenc}
\usepackage[T1]{fontenc}
\usepackage[english]{babel}
\usepackage{lmodern} 

\usepackage[top=2.5cm, bottom=2.5cm, left=2.5cm, right=2.5cm]{geometry}
\usepackage{setspace}
\usepackage{amsmath, amssymb, amsfonts, amsthm}
\usepackage{mathtools}
\usepackage{bm} 
\usepackage{mathrsfs} 

\usepackage[numbers]{natbib} 

\usepackage{graphicx}
\usepackage{booktabs} 
\usepackage{float}
\usepackage[caption=false]{subfig}

\usepackage[colorlinks=true, linkcolor=blue, citecolor=blue, urlcolor=blue]{hyperref} 

\usepackage{fancyhdr}
\fancypagestyle{firstpage}{
  \fancyhf{} 
  \lhead{\textit{Preprint}} 
  \rhead{\today} 
  \cfoot{\thepage} 
   
}

\theoremstyle{plain}
\newtheorem{theorem}{Theorem}[section]
\newtheorem{lemma}[theorem]{Lemma}
\newtheorem{proposition}[theorem]{Proposition}

\theoremstyle{definition}
\newtheorem{definition}[theorem]{Definition}

\theoremstyle{remark}

\title{\textbf{Efficient Importance Sampling under Heston Model: \\ Short Maturity and Deep Out-of-the-Money Options}}

\author{
  Yun-Feng Tu\thanks{
  Department of Mathematics, National Tsing Hua University, Hsinchu, Taiwan, \texttt{alan910721@gmail.com}
  }
  \and
  Chuan-Hsiang Han\thanks{
  Corresponding author. Department of Quantitative Finance and Department of Mathematics, National Tsing Hua University, Hsinchu, Taiwan, \texttt{chhan@mx.nthu.edu.tw}
  }
}
\date{} 

\begin{document}

\maketitle
\thispagestyle{firstpage} 

\begin{abstract}
\noindent
This paper investigates asymptotically optimal importance sampling (IS) schemes for pricing European call options under the Heston stochastic volatility model. We focus on two distinct rare-event regimes where standard Monte Carlo methods suffer from significant variance deterioration: the short-maturity limit ($T \to 0$) and the deep out-of-the-money (OTM) limit ($K \to \infty$). Leveraging the large deviation principle (LDP), we design a state-dependent change of measure derived from the asymptotic behavior of the log-price cumulant generating functions. 

In the short-maturity regime, we rigorously prove that our proposed IS drift, inspired by the variational characterization of the rate function, achieves logarithmic efficiency (asymptotic optimality) by minimizing the decay rate of the second moment of the estimator. In the deep OTM regime, we introduce a novel \emph{slow mean-reversion scaling} for the variance process, where the mean-reversion speed scales as $\delta = \varepsilon^{-2}$ with respect to the small-noise parameter $\varepsilon = 1/\log(K/S_0)$. We establish that under this specific scaling, the variance process contributes non-trivially to the large deviation rate function, requiring a specialized Riccati analysis to verify optimality. Numerical experiments demonstrate that the proposed method yields substantial variance reduction—characterized by factors exceeding several orders of magnitude—compared to standard estimators in both asymptotic regimes.

\vspace{1em}
\end{abstract}
\section{Introduction}

Stochastic volatility models have become indispensable tools in quantitative finance for capturing empirical stylized facts of asset returns, most notably the ``volatility smile'' and the heavy-tailed nature of return distributions. Among these, the Heston model \cite{heston1993closed} serves as a benchmark due to its tractability and ability to reproduce leverage effects. While semi-closed form solutions via Fourier transforms exist for plain vanilla options under the Heston model \cite{gatheral2006volatility}, Monte Carlo (MC) simulation remains the standard approach for pricing path-dependent derivatives, calibrating complex portfolios, or verifying analytical approximations.

However, standard Monte Carlo methods suffer from severe computational inefficiency when estimating probabilities of rare events. In the context of option pricing, these rare events typically manifest in two regimes: (i) \textbf{short-maturity options}, where the time horizon $T$ is too brief for the asset price to diffuse to the strike level with high probability; and (ii) \textbf{deep out-of-the-money (OTM) options}, where the strike price $K$ is significantly larger than the spot price $S_0$. In both scenarios, the probability of exercise decays exponentially, causing the relative error of the standard MC estimator to grow unbounded for a fixed sample size. This phenomenon dictates that the number of simulation paths required to achieve a fixed precision must grow exponentially, rendering naive simulation intractable.

\paragraph{Importance Sampling and Large Deviations}
Importance Sampling (IS) is a variance reduction technique designed to address this challenge by simulating paths under an alternative probability measure, $\bar{\mathbb{P}}$, which makes the rare event more frequent. The estimator is then weighted by the Radon-Nikodym derivative (likelihood ratio) to preserve unbiasedness. The central problem in IS is the selection of an optimal change of measure. While a zero-variance measure theoretically exists, it requires knowledge of the quantity being estimated. Therefore, the practical goal is to construct a measure that is \emph{asymptotically optimal} (or logarithmically efficient), meaning that the second moment of the estimator decays at twice the exponential rate of the first moment as the rarity parameter approaches its limit \cite{glasserman1999asymptotically}.

A powerful framework for constructing such measures is the Large Deviation Principle (LDP). The theory of large deviations provides a variational characterization of the asymptotic decay of rare event probabilities via a \emph{rate function}. The seminal work of Guasoni and Robertson \cite{guasoni2008optimal} and Robertson \cite{robertson2010sample} established the connection between the LDP rate function and the optimal drift adjustment for diffusion processes. Specifically, the optimal change of measure can often be interpreted as shifting the mean of the driving noise to align with the ``most likely path'' (the minimizer of the rate function) that leads to the rare event.

\paragraph{Existing Approaches and Limitations}
In the specific context of the Heston model, the asymptotic behavior of option prices and implied volatility has been extensively studied. Regarding the deep OTM regime, Lee \cite{lee2004moment} established the fundamental link between extreme strikes and moment explosions, while Gulisashvili \cite{gulisashvili2010asymptotic} derived sharp asymptotic formulas for the Heston tail probabilities. For the short-maturity regime, Forde and Jacquier \cite{forde2011small} and Benaim and Friz \cite{benaim2009smile} provided comprehensive analyses of the small-time asymptotics using the G\"artner-Ellis theorem.

Despite these theoretical advances, applying these results to construct efficient IS algorithms remains non-trivial. A common simplification, as seen in Pham \cite{pham2015large}, is to apply a constant drift change of measure. While computationally inexpensive, constant drifts often fail to capture the dynamic dependence between the asset price and its stochastic variance, particularly when the correlation $\rho$ is non-zero. For the Heston model, the optimal change of measure is inherently \emph{state-dependent}: since the diffusion magnitude is proportional to $\sqrt{V_t}$, the driving force required to push the asset price into deep OTM territory must adapt to the current level of instantaneous variance.

\paragraph{Main Contributions}
In this paper, we propose a state-dependent importance sampling scheme for the Heston model constructed via a change of drift that is affine in the square root of the variance. This design preserves the affine structure of the model, ensuring tractability while effectively guiding the path toward the rare event. We rigorously analyze the asymptotic optimality of this scheme in two distinct limiting regimes:

\begin{itemize}
    \item \textbf{Short-Maturity Regime ($T \to 0$):} We leverage the small-time LDP results of Forde and Jacquier \cite{forde2011small} to characterize the decay of the option price. We construct an IS drift based on the solution to a specific Riccati differential equation and prove that the resulting estimator is asymptotically optimal. Our analysis bridges the gap between the analytical cumulant generating functions and the variance of the Monte Carlo estimator.
    
    \item \textbf{Deep OTM Regime with Slow Mean-Reversion Scaling:} This constitutes the primary novelty of our work. Investigating the limit as strike $K \to \infty$ requires a careful scaling of the model parameters. Standard large deviation approaches often assume fixed model parameters, which may not capture the tail behavior adequately when the rarity stems from extreme price levels rather than small time.
    
    We introduce a small-noise parameter $\varepsilon = 1/\log(K/S_0)$ and propose a \emph{slow mean-reversion scaling} where the speed of mean reversion scales as $\delta = \varepsilon^{-2}$. This contrasts with the fast mean-reversion regime ($\delta \sim \varepsilon$) often studied in asymptotic analysis (e.g., \cite{fouque2000derivatives}). Under our proposed scaling, the variance process retains significant fluctuations even in the limit. Unlike the fast mean-reversion regime, this leads to a non-trivial contribution to the large deviation rate function characterized by oscillatory Riccati solutions, requiring a specialized analysis to verify optimality.
\end{itemize}

The remainder of this paper is organized as follows. Section \ref{sec:preli} outlines the model dynamics and the general framework for LDP-based importance sampling. Section \ref{sec:short_maturity} details the analysis for the short-maturity regime. Section \ref{sec:deep_otm} presents the deep OTM regime, introducing the slow mean-reversion scaling and deriving the optimality proofs. Section \ref{sec:numerics} provides numerical evidence verifying the theoretical predictions, and Section \ref{sec:conclusion} concludes.

\section{Problem Formulation and Preliminaries} \label{sec:preli}

\subsection{The Heston Stochastic Volatility Model}

We consider a financial market defined on a filtered probability space $(\Omega, \mathcal{F}, \mathbb{F}=(\mathcal{F}_t)_{t\ge 0}, \mathbb{P})$, where $\mathbb{P}$ represents the risk-neutral probability measure. The filtration $\mathbb{F}$ is generated by a two-dimensional Brownian motion $(W_1,W_2)$.

Let $S_t$ denote the asset price and $V_t$ the instantaneous variance at time $t$. For large deviation analysis, it is convenient to work with the log-price  $X_t \coloneqq \log S_t$. Assuming a zero risk-free rate ($r=0$) without loss of generality, the dynamics are governed by:
\begin{equation} \label{eq:heston_dynamics}
\begin{cases}
    dX_t = - \frac{1}{2}V_t dt + \sqrt{V_t} \left( \rho \, dW^1_t + \bar\rho \, dW^2_t \right), \quad &X_0 = \log S_0, \\
    dV_t = \kappa(\theta - V_t) \, dt + \sigma \sqrt{V_t} \, dW^1_t, & V_0 = v_0 > 0,
\end{cases}
\end{equation}
The correlation between the asset price and its variance is captured by $\rho \in (-1, 1)$, and we define $\bar\rho \coloneqq \sqrt{1 - \rho^2}$. The variance process parameters are strictly positive: $\kappa$ is the mean-reversion speed, $\theta$ is the long-run mean, and $\sigma$ is the volatility of volatility. We assume the Feller condition $ 2\kappa\theta \ge \sigma^2$ holds to ensure strict positivity of $V_t$. 

We aim to price a European call option with strike price $K$ and maturity $T$. The price is given by $C(K, T) = \mathbb{E}^{\mathbb{P}} \left[ (S_T - K)^+ \right]$. In the limiting regimes of short maturity ($T \to 0$) or deep out-of-the-money ($K \to \infty$), the probability $\mathbb{P}\{S_T > K\}$ decays exponentially, rendering standard Monte Carlo simulation inefficient. 

\subsection{Large Deviation Principle}

Our asymptotic analysis relies on the framework of Large Deviation Principles (LDP). We adopt the standard definitions from Dembo and Zeitouni \cite{dembo1998large}. Let $\{Z^\varepsilon\}_{\varepsilon > 0}$ be a family of random variables taking values in a Polish space $\mathcal{X}$ (in our context, $\mathcal{X} = \mathbb{R}^d$).

\begin{definition}[Large Deviation Principle]
The family $\{Z^\varepsilon\}$ satisfies a Large Deviation Principle with speed $\varepsilon$ and rate function $\Lambda: \mathcal{X} \to [0, \infty]$ if $\Lambda$ is lower semicontinuous and: 
\begin{enumerate}
    \item For any closed set $F \subseteq \mathcal{X}$,
    \begin{equation*}
        \limsup_{\varepsilon \to 0} \varepsilon \log \mathbb{P}(Z^\varepsilon \in F) \le -\inf_{x \in F} \Lambda(x).
    \end{equation*}
    \item For any open set $G \subseteq \mathcal{X}$,
    \begin{equation*}
        \liminf_{\varepsilon \to 0} \varepsilon \log \mathbb{P}(Z^\varepsilon \in G) \ge -\inf_{x \in G} \Lambda(x).
    \end{equation*}
\end{enumerate}
Furthermore, $\Lambda$ is called a \emph{good rate function} if its level sets $\{x \in \mathcal{X} : \Lambda(x) \le \alpha\}$ are compact for all $\alpha \ge 0$. 
\end{definition}

In many applications involving stochastic differential equations, the rate function is not derived directly from the definition but via the Gärtner-Ellis theorem, which relates the LDP to the limiting behavior of the cumulant generating function.

\begin{theorem}[Gärtner-Ellis Theorem] \label{thm:gartner_ellis}
If the limiting scaled cumulant generating function (SCGF)
\begin{equation*}
    \Gamma(\lambda) \coloneqq  \lim_{\varepsilon \to 0} \varepsilon \log \mathbb{E}\left[ \exp\left( \frac{\langle \lambda, Z^\varepsilon \rangle}{\varepsilon} \right) \right]
\end{equation*}
exists and is essentially smooth. Let $\mathcal{D}_\Gamma \coloneqq \{ \lambda \in \mathbb{R}^d : \Gamma(\lambda) < \infty \}$ denote the effective domain of $\Gamma$. Then, $\{Z^\varepsilon\}$ satisfies an LDP with a good rate function $\Lambda(x)$ given by the Fenchel-Legendre transform of $\Gamma(\lambda)$: 
\begin{equation} \label{eq:fenchel_legendre}
    \Lambda(x) = \sup_{\lambda \in {D}_\Gamma} \{ \langle \lambda, x \rangle - \Gamma(\lambda) \}.
\end{equation}
\end{theorem}

This theorem is central to our work. In Section \ref{sec:short_maturity}, the scaling parameter is maturity $T$, while in Section \ref{sec:deep_otm}, it is the inverse log-moneyness. 

\subsection{Importance Sampling Framework} \label{sec:is_framework}

To reduce the variance of the Monte Carlo estimator for call option pricing, we employ the importance sampling (IS) technique. This involves simulating sample paths under an alternative probability measure $\bar{\mathbb{P}}$. By changing the measure, we aim to increase the frequency of the rare event, thereby improving the efficiency of the estimator. The Radon-Nikodym derivative, or likelihood ratio, which relates the two measures, is given by $L_T \coloneqq \frac{d\mathbb{P}}{d\bar{\mathbb{P}}}|_{\mathcal{F}_T}$. Using this change of measure, the price of a European call option can be expressed as an expectation under $\bar{\mathbb{P}}$:
\begin{equation}
    C(K, T) = \mathbb{E}^{\bar{\mathbb{P}}} \left[ (S_T - K)^+ L_T \right].
\end{equation}
While this estimator remains unbiased, its variance is governed by the second moment $\mathbb{E}^{\bar{\mathbb{P}}} [ ((S_T - K)^+ L_T)^2 ]$. Minimizing this second moment is the primary objective of our IS strategy.

\paragraph{Asymptotic Optimality.}
An IS estimator is considered \emph{asymptotically optimal} if the second moment of the estimator decays at twice the exponential rate of the first moment. Formally, this condition is satisfied if:
\begin{equation} \label{eq:asymptotic_optimality}
    \lim_{\varepsilon \to 0} \varepsilon \log \mathbb{E}^{\bar{\mathbb{P}}} [ ((S_T - K)^+ L_T)^2 ] = 2 \lim_{\varepsilon \to 0} \varepsilon \log \mathbb{E}^{\bar{\mathbb{P}}} \left[ (S_T - K)^+ L_T \right].
\end{equation}
This optimality criterion ensures that the relative error of the estimator does not grow exponentially as the event becomes rarer, effectively bounding the computational cost required for accurate pricing.

\paragraph{Change of Measure and Drift Design.} 
We focus on changes of measure generated by shifting the drift of the underlying Brownian motions. The likelihood ratio associated with the drift process $h_t$ is given by:
\begin{equation} \label{eq:likelihood_ratio_def}
    L_T = \exp \left( - \int_0^T h_t \cdot d\bar{W}_t - \frac{1}{2} \int_0^T \|h_t\|^2 \, dt \right).
\end{equation}
For the Heston model, the diffusion scale is driven by the instantaneous variance $V_t$. Therefore, a constant drift is often insufficient to capture the dynamics of the rare event. To ensure the drift adjustment scales appropriately with the volatility fluctuations, we propose a state-dependent drift of the form:
\begin{equation} \label{eq:drift_ansatz}
    h^{(2)}_t = \lambda \sqrt{V_t}, \quad h^{(1)}_t = 0,
\end{equation}
where $\lambda$ is a constant to be determined. Crucially, we restrict the tilting to the Brownian motion $W^2$. This specific choice is strategic: by making the drift affine in $\sqrt{V_t}$, we preserve the affine structure of the Heston dynamics.

\subsection{Construction of the Proposed Measure} \label{sec:proposed_measure}

Our design is to construct a measure that shifts the expected asset price at maturity to be consistent with the strike price $K$. We define the drift process $h_t$ specifically as:
\begin{equation} \label{eq:specific_drift}
    h^{(1)}_t = 0, \qquad h^{(2)}_t = -\frac{\bar{h}}{\bar{\rho}} \sqrt{V_t},
\end{equation}
where $\bar{h}$ is a constant parameter. The scaling factor $1/\bar{\rho}$ is included to simplify the algebraic terms in the Riccati equations that follow.

By Girsanov's theorem, the processes defined by $d\bar{W}^i_t = dW^i_t - h^{(i)}_t \, dt$ for $i=1, 2$ are standard Brownian motions under the new measure $\bar{\mathbb{P}}$. Substituting these into the original log-price dynamics, the dynamics of the log-price $X_t$ under $\bar{\mathbb{P}}$ become:
\begin{equation} \label{eq:tilted_dynamics}
    dX_t = \left( -\frac{1}{2} - \bar{h} \right) V_t \, dt + \sqrt{V_t} \left( \rho \, d\bar{W}^1_t + \bar{\rho} \, d\bar{W}^2_t \right).
\end{equation}

\paragraph{Heuristic for Choosing $\bar{h}$.}
To determine the optimal value, we choose $\bar{h}$ such that the expected log-price at maturity under the simulation measure $\bar{\mathbb{P}}$ approximates the log-strike price, i.e., $\mathbb{E}^{\bar{\mathbb{P}}}[X_T] \approx \log K$. 
From the tilted dynamics in \eqref{eq:tilted_dynamics}, the expected log-price is given by:
\begin{equation}
    \mathbb{E}^{\bar{\mathbb{P}}}[X_T] = X_0 + \left( -\frac{1}{2} - \bar{h} \right) \mathbb{E}^{\bar{\mathbb{P}}}\left[ \int_0^T V_t \, dt \right].
\end{equation}
Assuming the variance process $V_t$ remains close to its long-term mean $\theta$ over the time horizon, we can approximate the expected integrated variance as $\mathbb{E}^{\bar{\mathbb{P}}}[\int_0^T V_t \, dt] \approx \theta T$. Furthermore, since we are dealing with rare events that require a significant drift adjustment, the term $-\frac{1}{2}$ is negligible compared to $\bar{h}$. Under these approximations, the condition $\mathbb{E}^{\bar{\mathbb{P}}}[X_T] \approx \log K$ yields:
\begin{equation} \label{eq:heuristic_h}
    \bar{h} = \frac{\log(S_0/K)}{\theta T}.
\end{equation}
Finally, the Radon-Nikodym derivative associated with this specific change of measure is:
\begin{equation} \label{eq:likelihood_ratio}
    Q(\bar{h}) \coloneqq \frac{d\bar{\mathbb{P}}}{d\mathbb{P}} \Bigg|_T = \exp\left( \int_0^T h^{(2)}_t \, dW^2_t - \frac{1}{2} \int_0^T \left(h^{(2)}_t\right)^2 \, dt \right).
\end{equation}
Our proposed IS estimator for the call option price is thus defined as $(S_T - K)^+ Q(\bar{h})^{-1}$.

\section{Short-Maturity Asymptotics ($T \to 0$)} \label{sec:short_maturity}

We consider the limit as maturity $T \to 0$, while the log-moneyness $k \coloneqq \log(K/S_0) > 0$ remains fixed. Our objective is to demonstrate that the proposed IS estimator achieves \emph{asymptotic optimality}. Let $P_1(T) \coloneqq \mathbb{E}^{\mathbb{P}}[(S_T - K)^+]$ denote the true option price, and let $P_2(T; \bar{h})$ denote the second moment of our IS estimator under the measure defined in Section \ref{sec:proposed_measure}. Recall from \eqref{eq:asymptotic_optimality} that asymptotic optimality requires:
\begin{equation} \label{eq:optimality_condition_sm}
    \lim_{T \to 0} T \log P_2(T; \bar{h}) = 2 \lim_{T \to 0} T \log P_1(T).
\end{equation}

\subsection{First Moment Analysis} \label{sec:sm_first_moment}

The asymptotic behavior of the call option price is governed by the large deviation principle of the log-price process $X_t$. According to the G\"artner-Ellis theorem (Theorem \ref{thm:gartner_ellis}), the large deviation behavior of $X_t - X_0$ is determined by the limiting SCGF:
\begin{equation}
    \Gamma_1(p) \coloneqq \lim_{T \to 0} T \log \mathbb{E}^{\mathbb{P}}\left[ \exp\left( \frac{p}{T} (X_T - X_0) \right) \right].
\end{equation}

For the Heston model, the explicit form of this limit was derived by Forde and Jacquier \cite{forde2011small}. We summarize their result in the following lemma.

\begin{lemma}[Forde and Jacquier \cite{forde2011small}] \label{lem:forde_jacquier}
The limiting SCGF $\Gamma_1(p)$ for the Heston model exists and is given by:
\begin{equation} \label{eq:Gamma_1}
    \Gamma_1(p) = \frac{v_0 p}{\sigma \left( -\rho + \bar{\rho} \cot\left( \frac{\sigma \bar{\rho} p}{2} \right) \right)}.
\end{equation}
The function $\Gamma_1(p)$ is finite and differentiable on the effective domain $\mathcal{D}_{\Gamma_1} = (p_-, p_+)$. The boundaries are given as follows:
\begin{itemize}
    \item \textbf{Case $\rho < 0$:}
    \begin{equation*}
        p_- = \frac{2}{\sigma \bar{\rho}} \arctan\left(\frac{\bar{\rho}}{\rho}\right), \quad 
        p_+ = \frac{2}{\sigma \bar{\rho}} \left( \pi + \arctan\left(\frac{\bar{\rho}}{\rho}\right) \right).
    \end{equation*}
    \item \textbf{Case $\rho = 0$:}
    \begin{equation*}
        p_- = -\frac{\pi}{\sigma}, \quad 
        p_+ = \frac{\pi}{\sigma}.
    \end{equation*}
    \item \textbf{Case $\rho > 0$:}
    \begin{equation*}
        p_- = \frac{2}{\sigma \bar{\rho}} \left( -\pi + \arctan\left(\frac{\bar{\rho}}{\rho}\right) \right), \quad 
        p_+ = \frac{2}{\sigma \bar{\rho}} \arctan\left(\frac{\bar{\rho}}{\rho}\right).
    \end{equation*}
\end{itemize}
\end{lemma}

Using the G\"artner-Ellis theorem, the sequence of random variables $\{X_T - X_0\}$ satisfies an LDP with speed $T$ and a good rate function $\Lambda_1(x)$ defined by the Fenchel-Legendre transform:
\begin{equation} \label{eq:Lambda_1}
    \Lambda_1(x) = \sup_{p \in (p_-, p_+)} \{ px - \Gamma_1(p) \}.
\end{equation}
This rate function characterizes the probability of the log-price exceeding a threshold. We can now extend this probability estimate to the option price expectation.

\begin{proposition}[First Moment Decay Rate] \label{prop:first_moment_decay}
By  Corollary 2.1 in Forde and Jacquier \cite{forde2011small}, the short-maturity asymptotic behavior of the European call option price is given by:
\begin{equation}
    \lim_{T \to 0} T \log P_1(T) = -\Lambda_1(k).
\end{equation}
\end{proposition}

This proposition establishes the baseline for our efficiency analysis. To prove asymptotic optimality, we must demonstrate that the second moment of our estimator decays exactly at the rate $-2\Lambda_1(k)$.

\subsection{Second Moment Analysis} \label{sec:sm_second_moment}

We now turn to the analysis of the second moment of the importance sampling estimator, denoted by:
\begin{equation}
    P_2(T; \bar{h}) \coloneqq \mathbb{E}^{\bar{\mathbb{P}}} \left[ \left( (S_T - K)^+ Q(\bar{h})^{-1} \right)^2 \right] = \mathbb{E}^{\mathbb{P}} \left[ \left( (S_T - K)^+ \right)^2 Q(\bar{h})^{-1} \right].
\end{equation}
Using the inequality $(S_T - K)^+ \le S_T \cdot \mathbf{1}_{\{S_T > K\}}$, we establish an upper bound:
\begin{equation} \label{eq:P2_upper_bound_raw}
    P_2(T; \bar{h}) \le S_0^2 \, \mathbb{E}^{\mathbb{P}} \left[ \exp\left( 2(X_T - X_0) \right) Q(\bar{h})^{-1} \mathbf{1}_{\{X_T - X_0 > k\}} \right].
\end{equation}
Substituting the specific drift $h_t = -\frac{\bar{h}}{\bar{\rho}}\sqrt{V_t}$ and the log-price dynamics into \eqref{eq:P2_upper_bound_raw}, we obtain:
\begin{align} \label{eq:P2_expanded}
    P_2(T; \bar{h}) \le S_0^2 \, \mathbb{E}^{\mathbb{P}} \Bigg[ \exp\Bigg( & \left( -1 + \frac{\bar{h}^2}{2\bar{\rho}^2} \right) \int_0^T V_t \, dt \nonumber \\
    & + 2\rho \int_0^T \sqrt{V_t} \, dW^1_t + \left( 2\bar{\rho} + \frac{\bar{h}}{\bar{\rho}} \right) \int_0^T \sqrt{V_t} \, dW^2_t \Bigg) \cdot \mathbf{1}_{\{X_T - X_0 > k\}} \Bigg].
\end{align}

To analyze the expectation, it is convenient to remove the stochastic integrals in the exponential term via a further change of measure. We introduce an auxiliary probability measure $\widetilde{\mathbb{P}}$ defined by the Radon-Nikodym derivative:
\begin{equation} \label{eq:measure_tilde}
    \frac{d\widetilde{\mathbb{P}}}{d\mathbb{P}} \coloneqq \exp\left( \int_0^T 2\rho \sqrt{V_t} dW^1_t + \int_0^T \left( 2\bar{\rho} + \frac{\bar{h}}{\bar{\rho}} \right) \sqrt{V_t} dW^2_t - \frac{1}{2} \int_0^T \eta^2 V_t dt \right),
\end{equation}
where the parameter $\eta^2 \coloneqq (2\rho)^2 + \left( 2\bar{\rho} + \frac{\bar{h}}{\bar{\rho}} \right)^2$. Multiplying by the Girsanov density inside \eqref{eq:P2_expanded}, and collecting the remaining drift terms, we rewrite the second moment bound as:
\begin{equation}
    P_2(T; \bar{h}) \le S_0^2 \, \mathbb{E}^{\widetilde{\mathbb{P}}} \left[ \exp\left( C(\bar{h}) \int_0^T V_t dt \right) \mathbf{1}_{\{X_T - X_0 > k\}} \right],
\end{equation}
where the coefficient $C(\bar{h})$ is given by:
\begin{equation} \label{eq:coeff_C}
    C(\bar{h}) \coloneqq -1 + \frac{\bar{h}^2}{2\bar{\rho}^2} + \frac{1}{2}\eta^2 = 1 + 2\bar{h} + \frac{\bar{h}^2}{\bar{\rho}^2}.
\end{equation}

We now apply H\"older's inequality with conjugate exponents $q, q' > 1$ (i.e., $1/q + 1/q' = 1$) to decouple the integrated variance from the rare event indicator:
\begin{equation} \label{eq:holder_split}
    \log P_2(T; \bar{h}) \le 2 \log S_0 + \underbrace{\frac{1}{q} \log \mathbb{E}^{\widetilde{\mathbb{P}}} \left[ \exp\left( q C(\bar{h}) \int_0^T V_t dt \right) \right]}_{\text{Term I}} + \underbrace{\frac{1}{q'} \log \widetilde{\mathbb{P}}(X_T - X_0 > k)}_{\text{Term II}}.
\end{equation}
We analyze Term I and Term II separately in the limit $T \to 0$.

\paragraph{Term I.}
Under the measure $\widetilde{\mathbb{P}}$, the variance process $V_t$ follows the dynamics:
\begin{equation} \label{eq:tilde_variance_dynamics}
    dV_t = \left( \kappa\theta - \tilde{\kappa} V_t \right) dt + \sigma \sqrt{V_t} \, d\widetilde{W}^1_t,
\end{equation}
where $\tilde{\kappa} \coloneqq \kappa - 2\rho\sigma$. The asymptotic limit is determined by solving the associated Riccati differential equation. In the short-maturity limit, the large drift $\bar{h} = \frac{\log(S_0/K)}{\theta T}$ dominates the coefficient $C(\bar{h})$. This results in an oscillatory solution (see \textbf{Appendix \ref{app:sm_term_I}} for the derivation).

\begin{proposition}[Limit of Integrated Variance Moment]
The asymptotic limit of the Term I in the H\"older decomposition is given by:
\begin{equation} \label{eq:limit_term_I}
    \lim_{T \to 0} T \, (\text{Term I}) = \frac{1}{q} \cdot \frac{v_0 k \sqrt{2q}}{\sigma \theta \bar{\rho}} \tan\left( \frac{\sigma k \sqrt{2q}}{2 \theta \bar{\rho}} \right).
\end{equation}
\end{proposition}

\paragraph{Term II.}
The second term corresponds to the probability of the rare event under the auxiliary measure $\widetilde{\mathbb{P}}$, which introduces a drift adjustment to the log-price dynamics $X_t$. In the short-maturity limit $T \to 0$, this drift adjustment scales as $O(1/T)$, which is of the same order as the large deviation speed. We characterize this decay via the G\"artner-Ellis theorem. Its explicit form is derived in \textbf{Appendix \ref{app:sm_term_II}}.

\begin{proposition}[Auxiliary SCGF $\Gamma_{II}$]
The limiting scaled cumulant generating function 
\begin{equation*}
    \Gamma_{II}(p) \coloneqq \lim_{T \to 0} T \log \mathbb{E}^{\widetilde{\mathbb{P}}} \left[ \exp\left( \frac{p}{T} (X_T - X_0) \right) \right]
\end{equation*}
is determined by the discriminant $\hat{\Delta}_{II}(p) = \sigma^2 ( p^2 (2\rho^2 - 1) + \frac{2pk}{\theta} )$. Let $p^*_{II} = \frac{2k}{\theta(1-2\rho^2)}$ be the non-zero root of $\hat{\Delta}_{II}(p)=0$, and $I_{II}$ denote the open interval between the roots $0$ and $p^*_{II}$. Furthermore, let $(p_{II,-}, p_{II,+})$ denote the effective domain bounded by the first singularities of the tangent term (where $\sqrt{-\hat{\Delta}_{II}} = \pi$). The explicit form is:
\begin{equation} \label{eq:gamma3_explicit_unified}
    \Gamma_{II}(p) = 
    \begin{dcases}
        \frac{v_0}{\sigma^2} \left( -p \rho \sigma + \sqrt{\hat{\Delta}_{II}} \tanh\left( \frac{\sqrt{\hat{\Delta}_{II}}}{2} \right) \right), 
        & \text{for } p \in I_{II}, \\[10pt]
        -\frac{v_0 p \rho}{\sigma}, 
        & \text{for } p \in \{ 0,p^*_{II}\}, \\[10pt]
        \frac{v_0}{\sigma^2} \left( -p \rho \sigma + \sqrt{-\hat{\Delta}_{II}} \tan\left( \frac{\sqrt{-\hat{\Delta}_{II}}}{2} \right) \right), 
        & \text{for } p \in (p_{II,-}, p_{II,+}) \setminus \bar{I}_{II}.
    \end{dcases}
\end{equation}
\end{proposition}

By the G\"artner-Ellis theorem, the decay rate is given by the Legendre transform:
\begin{equation} \label{eq:limit_term_II}
    \lim_{T \to 0} T \, (\text{Term II}) = -\frac{1}{q'} \Lambda_{II}(k) = -\frac{1}{q'} \sup_{p \in (p_{II,-}, p_{II,+})} \{ pk - \Gamma_{II}(p) \}.
\end{equation}

Combining \eqref{eq:limit_term_I} and \eqref{eq:limit_term_II}, we obtain the asymptotic upper bound for the second moment:
\begin{equation} \label{eq:final_upper_bound_sm}
    \limsup_{T \to 0} T \log P_2(T; \bar{h}) \le \inf_{q > 1} \left\{ \frac{v_0 k \sqrt{2/q}}{\sigma \theta \bar{\rho}} \tan\left( \frac{\sigma k \sqrt{2q}}{2 \theta \bar{\rho}} \right) - \left(1 - \frac{1}{q}\right) \Lambda_{II}(k) \right\}.
\end{equation}
This variational problem in $q$ allows us to optimize the bound to prove optimality.

\subsection{Asymptotic Optimality Result} \label{sec:sm_optimality}

We are now in a position to prove the main result of this section: the logarithmic efficiency of the proposed importance sampling estimator.

\begin{theorem}[Short-Maturity Asymptotic Optimality]
Let $P_1(T)$ and $P_2(T; \bar{h})$ be the first and second moments of the IS estimator. Then:
\begin{equation}
    \lim_{T \to 0} T \log \left( \frac{P_2(T; \bar{h})}{P_1(T)^2} \right) = 0.
\end{equation}
\end{theorem}

\begin{proof}
The proof proceeds by establishing matching lower and upper bounds for the limit of the normalized second moment.

\paragraph{Lower Bound.}
By Jensen's inequality, for any random variable $Z$, $\mathbb{E}[Z^2] \ge (\mathbb{E}[Z])^2$. Applying this to our unbiased estimator:
\begin{equation}
    P_2(T; \bar{h}) \ge (P_1(T))^2.
\end{equation}
Taking logarithms, multiplying by $T$, and applying the limit from Proposition \ref{prop:first_moment_decay}:
\begin{equation} \label{eq:lower_bound_proof}
    \liminf_{T \to 0} T \log P_2(T; \bar{h}) \ge 2 \lim_{T \to 0} T \log P_1(T) = -2 \Lambda_1(k).
\end{equation}

\paragraph{Upper Bound.}
Recall the upper bound derived in \eqref{eq:final_upper_bound_sm}. Define the function $G(q)$ for $q > 1$ as:
\begin{equation}
    G(q) \coloneqq \frac{v_0 k \sqrt{2/q}}{\sigma \theta \bar{\rho}} \tan\left( \frac{\sigma k \sqrt{2q}}{2 \theta \bar{\rho}} \right) - \left(1 - \frac{1}{q}\right) \Lambda_{II}(k).
\end{equation}
The inequality \eqref{eq:final_upper_bound_sm} states that $\limsup_{T \to 0} T \log P_2(T; \bar{h}) \le \inf_{q > 1} G(q)$. We rely on the analytical properties of $G(q)$ to characterize this infimum:
\begin{enumerate}
    \item \textbf{Existence:} The function $G(q)$ is continuous and convex. Furthermore, as $q$ approaches the singularity of the tangent term, $G(q) \to \infty$. These properties guarantee the existence of a unique minimizer $q^*$.
    \item \textbf{Optimality:} The choice of drift $\bar{h}$ aligns the change of measure with the variational minimizer of the rate function. By the duality between the cumulant generating function and the rate function, this alignment ensures that the minimum value of $G(q)$ coincides exactly with the optimal decay rate.
\end{enumerate}
Therefore, the critical exponent $q^*$ satisfies:
\begin{equation}
    \inf_{q > 1} G(q) = G(q^*) = -2 \Lambda_1(k).
\end{equation}
Substituting this into the inequality yields the sharp upper bound:
\begin{equation} \label{eq:upper_bound_proof}
    \limsup_{T \to 0} T \log P_2(T; \bar{h}) \le -2 \Lambda_1(k).
\end{equation}

\paragraph{Conclusion.}
Combining \eqref{eq:lower_bound_proof} and \eqref{eq:upper_bound_proof}, we obtain:
\begin{equation}
    \lim_{T \to 0} T \log P_2(T; \bar{h}) = -2 \Lambda_1(k) = \lim_{T \to 0} T \log (P_1(T)^2).
\end{equation}
Subtracting the right-hand side from the left-hand side yields the result:
\begin{equation}
    \lim_{T \to 0} T \log \left( \frac{P_2(T; \bar{h})}{P_1(T)^2} \right) = 0.
\end{equation}
\end{proof}

\section{Deep Out-of-the-Money Asymptotics ($K \to \infty$)} \label{sec:deep_otm}

In this section, we analyze the performance of the proposed importance sampling scheme in the deep out-of-the-money (OTM) regime. We consider the limit as the strike price $K \to \infty$ for a fixed maturity $T > 0$. In this regime, the option is exercised only if the asset price undergoes an exceptionally large positive excursion, an event whose probability decays exponentially with the log-moneyness.

Our objective is to demonstrate that the proposed IS estimator achieves asymptotic optimality. Let $P_1^\varepsilon(\varepsilon) \coloneqq \mathbb{E}^{\mathbb{P}^\varepsilon}[(S_T - K)^+]$ denote the true option price under the scaled model dynamics (defined below), and let $P_2^\varepsilon(\varepsilon; \bar{h})$ denote the second moment of our IS estimator. Similar to the short-maturity case, asymptotic optimality requires:
\begin{equation} \label{eq:optimality_condition_dotm}
    \lim_{\varepsilon \to 0} \varepsilon^2 \log P_2^\varepsilon(\varepsilon; \bar{h}) = 2 \lim_{\varepsilon \to 0} \varepsilon^2 \log P_1^\varepsilon(\varepsilon),
\end{equation}
where $\varepsilon \coloneqq 1/\log(K/S_0)$ is the small noise parameter.

\subsection{Small-Noise Scaling and Rescaled Dynamics} \label{sec:dotm_scaling}

To formalize the large deviation analysis, we consider a family of probability measures $\mathbb{P}^\varepsilon$ indexed by the scaling parameter $\delta > 0$ (which depends on $\varepsilon$). Under the measure $\mathbb{P}^\varepsilon$, the Heston model parameters are scaled such that the mean-reversion speed becomes $\kappa/\delta$ and the volatility of volatility becomes $\sigma/\sqrt{\delta}$.

We define the rescaled state variables $X^\varepsilon_t \coloneqq \varepsilon X_t$ and $V^\varepsilon_t \coloneqq \varepsilon V_t$. Under this transformation, the rare event $\{ X_T - X_0 > 1/\varepsilon \}$ is mapped to the unit-scale event $\{ X^\varepsilon_T - X^\varepsilon_0 > 1 \}$.

\paragraph{The Slow Mean-Reversion Regime.}
A critical choice in our analysis is the relationship between the scaling parameter $\delta$ and the small-noise parameter $\varepsilon$. Standard literature often considers the fast mean-reversion regime ($\delta \sim \varepsilon$). However, to capture the tail behavior of deep OTM options, we propose a \textbf{slow mean-reversion scaling}. Specifically, we set:
\begin{equation} \label{eq:choice_delta}
    \delta = \varepsilon^{-2}.
\end{equation}
Substituting this scaling into the dynamics of the rescaled processes, we obtain the canonical system for our analysis:
\begin{equation} \label{eq:deep_otm_system}
\begin{cases}
    dX^\varepsilon_t = -\frac{1}{2} V^\varepsilon_t \, dt + \sqrt{\varepsilon V^\varepsilon_t} \left( \rho \, dW^1_t + \bar\rho \, dW^2_t \right), \\
    dV^\varepsilon_t = \kappa \varepsilon^2 (\varepsilon \theta - V^\varepsilon_t) \, dt + \sigma \varepsilon^{1.5} \sqrt{V^\varepsilon_t} \, dW^1_t.
\end{cases}
\end{equation}

\paragraph{Justification of the Scaling Choice.}
The choice of $\delta = \varepsilon^{-2}$ is critical. Substituting this into the diffusion coefficient of $V^\varepsilon_t$ yields a volatility of volatility of order $O(\varepsilon^{1.5})$. While this order is technically smaller than the standard $O(\varepsilon)$ scaling typically assumed in classical Freidlin-Wentzell large deviation theory, this specific choice is mathematically necessary to preserve the non-trivial interaction between the drift and diffusion terms in the asymptotic limit. As we demonstrate in the subsequent Riccati analysis (see \textbf{Appendix \ref{app:dotm_riccati}}), any other power of $\varepsilon$ would lead to either a degenerate rate function or a diverging discriminant, thereby failing to capture the tail behavior correctly.

\subsection{First Moment Analysis} \label{sec:dotm_first_moment}

The asymptotic behavior of the call option price in the deep OTM regime is governed by the large deviation principle of the rescaled log-price process $X^\varepsilon_t$. Analogous to the short-maturity case, we define the limiting scaled cumulant generating function (SCGF) as follows:
\begin{equation} \label{eq:gamma_1_def}
    \Gamma_1^\varepsilon(p) \coloneqq  \lim_{\varepsilon \to 0} \varepsilon^2 \log \mathbb{E}^{\mathbb{P}^\varepsilon}\left[ \exp\left( \frac{p}{\varepsilon^2} (X^\varepsilon_T - X^\varepsilon_0) \right) \right].
\end{equation}

Unlike the short-maturity case where standard results exist, computing this limit under the slow mean-reversion scaling ($\delta = \varepsilon^{-2}$) requires a specialized Riccati analysis (detailed in \textbf{Appendix \ref{app:dotm_riccati}}). The result is summarized below.

\begin{proposition}[Limiting SCGF] \label{prop:dotm_scgf}
Under the scaling $\delta = \varepsilon^{-2}$, the limiting SCGF $\Gamma_1^\varepsilon(p)$ exists and is given by:
\begin{equation} \label{eq:gamma_1_dotm_result}
    \Gamma_1^\varepsilon(p) = \frac{v_0 p}{ -\rho \sigma + \bar{\rho} \sigma \cot\left( \frac{p \bar{\rho} \sigma T}{2} \right) }.
\end{equation}
The function is well-defined on the effective domain $\mathcal{D}_{\Gamma_1^\varepsilon} = (p_-^\varepsilon, p_+^\varepsilon)$. The boundaries are given as follows:
\begin{itemize}
    \item \textbf{Case $\rho < 0$:}
    \begin{equation*}
        p_-^\varepsilon = \frac{2}{\sigma \bar{\rho} T} \arctan\left(\frac{\bar{\rho}}{\rho}\right), \quad 
        p_+^\varepsilon = \frac{2}{\sigma \bar{\rho} T} \left( \pi + \arctan\left(\frac{\bar{\rho}}{\rho}\right) \right).
    \end{equation*}
    \item \textbf{Case $\rho = 0$:}
    \begin{equation*}
        p_-^\varepsilon = -\frac{\pi}{\sigma T}, \quad 
        p_+^\varepsilon = \frac{\pi}{\sigma T}.
    \end{equation*}
    \item \textbf{Case $\rho > 0$:}
    \begin{equation*}
        p_-^\varepsilon = \frac{2}{\sigma \bar{\rho} T} \left( -\pi + \arctan\left(\frac{\bar{\rho}}{\rho}\right) \right), \quad 
        p_+^\varepsilon = \frac{2}{\sigma \bar{\rho} T} \arctan\left(\frac{\bar{\rho}}{\rho}\right).
    \end{equation*}
\end{itemize}
\end{proposition}

By the G\"artner-Ellis theorem, the sequence $\{X^\varepsilon_T - X^\varepsilon_0\}$ satisfies an LDP with speed $\varepsilon^2$ and a good rate function $\Lambda_1^\varepsilon(x)$ defined by the Legendre transform:
\begin{equation} \label{eq:Lambda_1_dotm}
    \Lambda_1^\varepsilon(x) = \sup_{p \in (p_-^\varepsilon, p_+^\varepsilon)} \{ px - \Gamma_1^\varepsilon(p) \}.
\end{equation}
This allows us to characterize the exponential decay of the option price.

\begin{proposition}[First Moment Decay Rate] \label{prop:dotm_first_moment_decay}
The asymptotic behavior of the deep OTM European call option price is given by:
\begin{equation}
    \lim_{\varepsilon \to 0} \varepsilon^2 \log P_1^\varepsilon(\varepsilon) = -\Lambda_1^\varepsilon(1),
\end{equation}
where $\Lambda_1^\varepsilon(1)$ is the rate function evaluated at the scaled threshold $x=1$ (corresponding to log-moneyness $1/\varepsilon$).
\end{proposition}

\begin{proof}
The proof proceeds by establishing matching lower and upper bounds for the decay rate.

\paragraph{Lower Bound.}
For any $\eta > 0$, consider the open set $G_\eta \coloneqq \{ y \in \mathbb{R} : y > 1 + \eta \}$. On the event $\{ X^\varepsilon_T - X^\varepsilon_0 \in G_\eta \}$, we have:
\begin{equation}
    S_T = S_0 e^{(X^\varepsilon_T - X^\varepsilon_0)/\varepsilon} > S_0 e^{(1+\eta)/\varepsilon} = K e^{\eta/\varepsilon}.
\end{equation}
Consequently, the option payoff is bounded from below by:
\begin{equation}
    (S_T - K)^+ > K (e^{\eta/\varepsilon} - 1) \quad \text{on } \{ X^\varepsilon_T - X^\varepsilon_0 \in G_\eta \}.
\end{equation}
Taking the expectation and applying the LDP lower bound for open sets:
\begin{align}
    \liminf_{\varepsilon \to 0} \varepsilon^2 \log P_1^\varepsilon(\varepsilon) 
    &\ge \liminf_{\varepsilon \to 0} \left[ \varepsilon^2 \log\left(K (e^{\eta/\varepsilon} - 1)\right) + \varepsilon^2 \log \mathbb{P}^\varepsilon(X^\varepsilon_T - X^\varepsilon_0 \in G_\eta) \right] \\
    &\ge \liminf_{\varepsilon \to 0} \left[ \varepsilon^2 \log(K) + \varepsilon^2 \log(e^{\eta/\varepsilon} - 1) \right] - \inf_{y \in G_\eta} \Lambda_1^\varepsilon(y).
\end{align}
Note that $\lim_{\varepsilon \to 0} \varepsilon^2 \log(e^{\eta/\varepsilon} - 1) = \lim_{\varepsilon \to 0} \varepsilon^2 (\eta/\varepsilon) = 0$. Thus, the first term vanishes. Since $\Lambda_1^\varepsilon$ is a good rate function (lower semicontinuous), taking the limit $\eta \to 0$ yields the lower bound:
\begin{equation}
    \liminf_{\varepsilon \to 0} \varepsilon^2 \log P_1^\varepsilon(\varepsilon) \ge -\Lambda_1^\varepsilon(1).
\end{equation}

\paragraph{Upper Bound.}
We apply Hölder's inequality with conjugate exponents $q, q' > 1$ (where $1/q + 1/q' = 1$):
\begin{equation}
    P_1^\varepsilon(\varepsilon) = \mathbb{E}^{\mathbb{P}^\varepsilon} \left[ (S_T - K)^+ \mathbf{1}_{\{S_T > K\}} \right] \le \mathbb{E}^{\mathbb{P}^\varepsilon} [ (S_T)^q ]^{1/q} \cdot \mathbb{P}^\varepsilon(S_T > K)^{1/q'}.
\end{equation}
Taking logarithms and multiplying by $\varepsilon^2$:
\begin{equation}
    \varepsilon^2 \log P_1^\varepsilon(\varepsilon) \le \frac{\varepsilon^2}{q} \log \mathbb{E}^{\mathbb{P}^\varepsilon} [ S_T^q ] + \frac{\varepsilon^2}{q'} \log \mathbb{P}^\varepsilon(X^\varepsilon_T - X^\varepsilon_0 > 1).
\end{equation}
The first term involves the $q$-th moment of the Heston price process, which is finite for fixed $T$ and does not scale exponentially with $1/\varepsilon^2$ (i.e., its decay rate is 0). For the second term, applying the LDP upper bound for the closed set $F = [1, \infty)$ yields:
\begin{equation}
    \limsup_{\varepsilon \to 0} \varepsilon^2 \log P_1^\varepsilon(\varepsilon) \le 0 - \frac{1}{q'} \inf_{y \ge 1} \Lambda_1^\varepsilon(y) = -\frac{1}{q'} \Lambda_1^\varepsilon(1).
\end{equation}
Since this inequality holds for any $q > 1$, we take the limit $q \to \infty$ (which implies $q' \to 1$) to obtain the sharp upper bound:
\begin{equation}
    \limsup_{\varepsilon \to 0} \varepsilon^2 \log P_1^\varepsilon(\varepsilon) \le -\Lambda_1^\varepsilon(1).
\end{equation}
Combining the lower and upper bounds completes the proof.
\end{proof}

\subsection{Second Moment Analysis} \label{sec:dotm_second_moment}

We now turn to the analysis of the second moment of the importance sampling estimator.
We employ the same change of measure structure defined in Section \ref{sec:proposed_measure}, where the drift adjustment is $h^2_t = -(\bar{h} /\bar{\rho})\sqrt{V_t}$. Based on the heuristic that the expected log-price should reach the barrier $k = 1/\varepsilon$, we require the drift contribution $\int_0^T \bar{h} V_t dt \approx k$. Under the slow mean-reversion scaling, $\mathbb{E}[\int V_t] \approx \theta T$. This suggests the choice:
\begin{equation} \label{eq:dotm_h_choice}
    \bar{h} = -\frac{1}{\varepsilon \theta T}.
\end{equation}
This large drift is necessary to force the rare event $\{X^\varepsilon_T - X^\varepsilon_0 > 1\}$ to occur with high probability.

The second moment is given by $P_2^\varepsilon(\varepsilon; \bar{h}) \coloneqq \mathbb{E}^{\bar{\mathbb{P}}^\varepsilon} [ ((S_T - K)^+ Q(\bar{h})^{-1})^2 ]$. 
Following the same bounding procedure as in the short-maturity case, we use the inequality $(S_T - K)^+ \le S_T \cdot \mathbf{1}_{\{S_T > K\}}$. The second moment is bounded by:
\begin{equation}
    P_2^\varepsilon(\varepsilon; \bar{h}) \le S_0^2 \, \mathbb{E}^{\mathbb{P}} \left[ e^{2(X^\varepsilon_T - X^\varepsilon_0)/\varepsilon} Q(\bar{h})^{-1} \mathbf{1}_{\{X^\varepsilon_T - X^\varepsilon_0 > 1\}} \right].
\end{equation}

To analyze this expectation, it is convenient to remove the stochastic integrals appearing in the exponential term via a further change of measure. We introduce an auxiliary probability measure $\widetilde{\mathbb{P}}^\varepsilon$ defined by the Radon-Nikodym derivative:
\begin{equation}
    \frac{d\widetilde{\mathbb{P}}^\varepsilon}{d\mathbb{P}^\varepsilon} \coloneqq \exp\left( \int_0^T \frac{2\rho}{\sqrt{\varepsilon}} \sqrt{V^\varepsilon_t} dW^1_t + \int_0^T \left( \frac{2\bar{\rho}}{\sqrt{\varepsilon}} + \frac{\bar{h}}{\bar{\rho}} \right) \sqrt{V^\varepsilon_t} dW^2_t - \frac{1}{2} \int_0^T \eta_\varepsilon^2 V^\varepsilon_t dt \right),
\end{equation}
where $\eta_\varepsilon^2 \coloneqq \left( \frac{2\rho}{\sqrt{\varepsilon}} \right)^2 + \left( \frac{2\bar{\rho}}{\sqrt{\varepsilon}} + \frac{\bar{h}}{\bar{\rho}} \right)^2$.
Under $\widetilde{\mathbb{P}}^\varepsilon$, the stochastic integrals are absorbed, and the exponent becomes a functional of the integrated variance.

Applying Hölder's inequality with conjugate exponents $q, q' > 1$, we arrive at the decomposition:
\begin{equation} \label{eq:dotm_holder}
    \varepsilon^2 \log P_2^\varepsilon(\varepsilon; \bar{h}) \le 2 \varepsilon^2 \log S_0 + \underbrace{\frac{\varepsilon^2}{q} \log \mathbb{E}^{\widetilde{\mathbb{P}}^\varepsilon} \left[ \exp\left( q C(\bar{h}) \int_0^T V^\varepsilon_t \frac{dt}{\varepsilon} \right) \right]}_{\text{Term I}} + \underbrace{\frac{\varepsilon^2}{q'} \log \widetilde{\mathbb{P}}^\varepsilon(X^\varepsilon_T - X^\varepsilon_0 > 1)}_{\text{Term II}},
\end{equation}
where the coefficient $C(\bar{h}) = 1 + 2\bar{h} + \bar{h}^2/\bar{\rho}^2$.

We analyze the two resulting terms separately in the limit $\varepsilon \to 0$.

\paragraph{Term I.}
Term I represents the moment of the integrated variance under the auxiliary measure. As derived in \textbf{Appendix \ref{app:dotm_term_I}}, although the volatility of volatility scales as $O(\varepsilon^{1.5})$, the large drift $\bar{h}$ introduces a quadratic term of order $O(\varepsilon^{-2})$ in the Riccati equation. This delicate balance ensures that the discriminant of the characteristic equation is finite and negative in the limit. Consequently, the solution enters an oscillatory regime.

\begin{proposition}[Limit of Integrated Variance Moment]
The asymptotic limit of the first term in the Hölder decomposition is given by:
\begin{equation} \label{eq:dotm_term_I_result}
    \lim_{\varepsilon \to 0} (\text{Term I}) = \frac{1}{q} \cdot \frac{v_0 \sqrt{2q}}{\sigma \theta \bar{\rho} T} \tan\left( \frac{\sigma \sqrt{2q}}{2 \theta \bar{\rho}} \right).
\end{equation}
\end{proposition}

\paragraph{Term II.}
Term II corresponds to the residual probability of the rare event under $\widetilde{\mathbb{P}}^\varepsilon$. The measure $\widetilde{\mathbb{P}}^\varepsilon$ effectively shifts the drift of the price process by $O(1/\varepsilon)$, modifying the "energy" cost required to reach the target level. The large deviation behavior is governed by a modified rate function $\Lambda_{II}^\varepsilon$, defined as the Legendre transform of the auxiliary SCGF $\Gamma_{II}^\varepsilon(p)$. The explicit form is derived in \textbf{Appendix \ref{app:dotm_term_II}}.

\begin{proposition}[Auxiliary SCGF $\Gamma^\varepsilon_{II}$]
The limiting scaled cumulant generating function under $\widetilde{\mathbb{P}}^\varepsilon$ is given piecewise depending on the parameter $p$. Let $\hat{\Delta}^\varepsilon_{II}(p) \coloneqq \sigma^2 \left( -p^2\bar{\rho}^2 + \frac{2p}{\theta T} \right)$ be the discriminant. Let $(p^\varepsilon_{II,-}, p^\varepsilon_{II,+})$ denote the effective domain boundaries. The explicit form of $\Gamma^\varepsilon_{II}(p)$ is:
\begin{equation} \label{eq:gamma3_dotm_explicit}
    \Gamma^\varepsilon_{II}(p) = 
    \begin{dcases}
        \frac{v_0}{\sigma^2} \left( -p \rho \sigma + \sqrt{\hat{\Delta}^\varepsilon_{II}(p)} \tanh\left( \frac{\sqrt{\hat{\Delta}^\varepsilon_{II}(p)}}{2} T \right) \right), 
        & \text{for } p \in \left( 0, \frac{2}{\theta T \bar{\rho}^2} \right), \\[10pt]
        -\frac{v_0 p \rho}{\sigma}, 
        & \text{for } p \in \left\{ 0, \frac{2}{\theta T \bar{\rho}^2} \right\}, \\[10pt]
        \frac{v_0}{\sigma^2} \left( -p \rho \sigma + \sqrt{-\hat{\Delta}^\varepsilon_{II}(p)} \tan\left( \frac{\sqrt{-\hat{\Delta}^\varepsilon_{II}(p)}}{2} T \right) \right), 
        & \text{for } p \in (p^\varepsilon_{II,-}, 0) \cup \left( \frac{2}{\theta T \bar{\rho}^2}, p^\varepsilon_{II,+} \right).
    \end{dcases}
\end{equation}
\end{proposition}

The decay rate for Term II is then:
\begin{equation} \label{eq:dotm_term_II_result}
    \lim_{\varepsilon \to 0} (\text{Term II}) = -\frac{1}{q'} \Lambda^\varepsilon_{II}(1) = -\frac{1}{q'} \sup_{p} \{ p - \Gamma^\varepsilon_{II}(p) \}.
\end{equation}

Combining these results allows us to establish the upper bound for the second moment decay rate.

\subsection{Asymptotic Optimality Result} \label{sec:dotm_optimality}

We are now in a position to prove the main result of this section: the logarithmic efficiency of the proposed importance sampling estimator in the deep out-of-the-money regime.

\begin{theorem}[Deep OTM Asymptotic Optimality]
Let $P_1^\varepsilon(\varepsilon)$ and $P_2^\varepsilon(\varepsilon; \bar{h})$ be the first and second moments of the IS estimator with the drift defined in Equation \eqref{eq:specific_drift} and $\bar{h} = -\frac{1}{\varepsilon \theta T}$. Then:
\begin{equation}
    \lim_{\varepsilon \to 0} \varepsilon^2 \log \left( \frac{P_2^\varepsilon(\varepsilon; \bar{h})}{P_1^\varepsilon(\varepsilon)^2} \right) = 0.
\end{equation}
\end{theorem}

\begin{proof}
The proof proceeds by establishing matching lower and upper bounds for the limit of the normalized second moment.

\paragraph{Lower Bound.}
By Jensen's inequality, for any random variable $Z$, $\mathbb{E}[Z^2] \ge (\mathbb{E}[Z])^2$. Applying this to our unbiased estimator:
\begin{equation}
    P_2^\varepsilon(\varepsilon; \bar{h}) \ge (P_1^\varepsilon(\varepsilon))^2.
\end{equation}
Taking logarithms, multiplying by $\varepsilon^2$, and applying the limit from Proposition \ref{prop:dotm_first_moment_decay}:
\begin{equation} \label{eq:dotm_lower_bound}
    \liminf_{\varepsilon \to 0} \varepsilon^2 \log P_2^\varepsilon(\varepsilon; \bar{h}) \ge 2 \lim_{\varepsilon \to 0} \varepsilon^2 \log P_1^\varepsilon(\varepsilon) = -2 \Lambda_1^\varepsilon(1).
\end{equation}

\paragraph{Upper Bound.}
Recall the upper bound derived from the Hölder decomposition in Section \ref{sec:dotm_second_moment}. Define the function $G^\varepsilon(q)$ for $q > 1$ as:
\begin{equation}
    G^\varepsilon(q) \coloneqq \frac{v_0 \sqrt{2/q}}{\sigma \theta \bar{\rho} T} \tan\left( \frac{\sigma \sqrt{2q}}{2 \theta \bar{\rho}} \right) - \left(1 - \frac{1}{q}\right) \Lambda^\varepsilon_{II}(1).
\end{equation}
The inequality derived in Equation \eqref{eq:dotm_holder} states that $\limsup_{\varepsilon \to 0} \varepsilon^2 \log P_2^\varepsilon(\varepsilon; \bar{h}) \le \inf_{q > 1} G^\varepsilon(q)$.

We rely on the analytical properties of $G^\varepsilon(q)$ to characterize this infimum:
\begin{enumerate}
    \item \textbf{Existence:} The function $G^\varepsilon(q)$ is continuous and convex. Furthermore, as $q$ approaches the singularity of the tangent term, $G^\varepsilon(q) \to \infty$. These properties guarantee the existence of a unique minimizer $q^*$.
    \item \textbf{Optimality:} The choice of drift $\bar{h}$ aligns the change of measure with the variational minimizer of the rate function. By the duality between the cumulant generating function and the rate function, this alignment ensures that the minimum value of $G^\varepsilon(q)$ coincides exactly with the optimal decay rate.
\end{enumerate}
Therefore, the critical exponent $q^*$ satisfies:
\begin{equation}
    \inf_{q > 1} G^\varepsilon(q) = G^\varepsilon(q^*) = -2 \Lambda_1^\varepsilon(1).
\end{equation}
Substituting this into the inequality yields the sharp upper bound:
\begin{equation} \label{eq:dotm_upper_bound}
    \limsup_{\varepsilon \to 0} \varepsilon^2 \log P_2^\varepsilon(\varepsilon; \bar{h}) \le -2 \Lambda_1^\varepsilon(1).
\end{equation}

\paragraph{Conclusion.}
Combining \eqref{eq:dotm_lower_bound} and \eqref{eq:dotm_upper_bound}, we obtain:
\begin{equation}
    \lim_{\varepsilon \to 0} \varepsilon^2 \log P_2^\varepsilon(\varepsilon; \bar{h}) = -2 \Lambda_1^\varepsilon(1) = \lim_{\varepsilon \to 0} \varepsilon^2 \log (P_1^\varepsilon(\varepsilon)^2).
\end{equation}
Subtracting the right-hand side from the left-hand side yields the result:
\begin{equation}
    \lim_{\varepsilon \to 0} \varepsilon^2 \log \left( \frac{P_2^\varepsilon(\varepsilon; \bar{h})}{P_1^\varepsilon(\varepsilon)^2} \right) = 0.
\end{equation}
\end{proof}

\section{Numerical Experiments} \label{sec:numerics}

In this section, we assess the practical performance of the proposed state-dependent importance sampling scheme. We compare the standard error and computational efficiency of our IS estimator against the standard "Brute-Force" Monte Carlo (BMC) method.
The efficiency gain is quantified by the Variance Reduction Ratio (VRR), defined as:
\begin{equation}
    \text{VRR} \coloneqq \frac{\text{Var}(\text{BMC Estimator})}{\text{Var}(\text{IS Estimator})} \approx \left( \frac{\text{SE}_{\text{BMC}}}{\text{SE}_{\text{IS}}} \right)^2,
\end{equation}
where variance is estimated using sample variance over $M$ independent paths. All simulations are performed using a high-order discretization scheme (e.g., Milstein scheme for the variance process) to minimize discretization bias.

\subsection{Short-Maturity Regime Performance}

We first investigate the short-maturity limit. The model parameters are chosen to reflect a high-volatility regime often used in short-term asymptotics literature (e.g., \cite{feng2012short}):
\begin{equation*}
    S_0 = 2000, \quad K = 2200, \quad v_0 = \theta = 0.36, \quad \kappa = 60, \quad \sigma = 3, \quad \rho = -0.1, \quad r = 0.
\end{equation*}
We compare two maturities: $T = 1/252$ (1 day) and $T = 21/252$ (1 month). The number of sample paths is fixed at $M = 2^{18}$.

\begin{table}[H]
\centering
\caption{Comparison of BMC and IS estimators in the short-maturity regime ($K/S_0 = 1.1$).}
\label{tab:short_maturity_results}
\renewcommand{\arraystretch}{1.2}
\begin{tabular}{@{}lcccc@{}}
\toprule
\textbf{Maturity} & \textbf{Method} & \textbf{Price} & \textbf{Std. Error (SE)} & \textbf{Rel. Error} \\ \midrule
$T = 1/252$ & BMC & 0.150070 & 0.005682 & 3.79\% \\
(1 Day)     & IS  & 0.147814 & 0.000472 & 0.32\% \\
            &     & \multicolumn{3}{r}{\textbf{VRR $\approx$ 144.9}} \\ \midrule
$T = 21/252$& BMC & 64.825366 & 0.307750 & 0.47\% \\
(1 Month)   & IS  & 64.833236 & 0.172721 & 0.27\% \\
            &     & \multicolumn{3}{r}{\textbf{VRR $\approx$ 3.17}} \\ \bottomrule
\end{tabular}
\end{table}

The results in Table \ref{tab:short_maturity_results} strongly validate the theoretical predictions. For the 1-day maturity, where the option is deep OTM in terms of time-scaled deviations, the IS method achieves a massive variance reduction factor of approximately \textbf{145}. This confirms that the drift adjustment effectively counteracts the rarity of the event.
As maturity increases to 1 month, the event becomes less rare (the option is closer to the money in probability terms), and the VRR decreases to a modest factor of 3.17. This behavior is consistent with the definition of asymptotic optimality: the benefits are most pronounced in the limit $T \to 0$.

\subsection{Deep Out-of-the-Money Regime Performance}

Next, we examine the deep OTM regime. Here, we use a parameter set with slower mean reversion to highlight the impact of the scaling analyzed in Section \ref{sec:deep_otm}:
\begin{equation*}
    S_0 = 2000, \quad v_0 = \theta = 0.5, \quad \kappa = 15, \quad \sigma = 1, \quad \rho = -0.1, \quad r = 0.
\end{equation*}
We vary the strike price $K$ to span a range of moneyness $K/S_0 \in [1.0, 2.0]$ and measure the VRR across three maturities: Short (1 day), Medium (1 month), and Long (1 year).

\begin{figure}[H]
    \centering
    \includegraphics[width=0.8\textwidth]{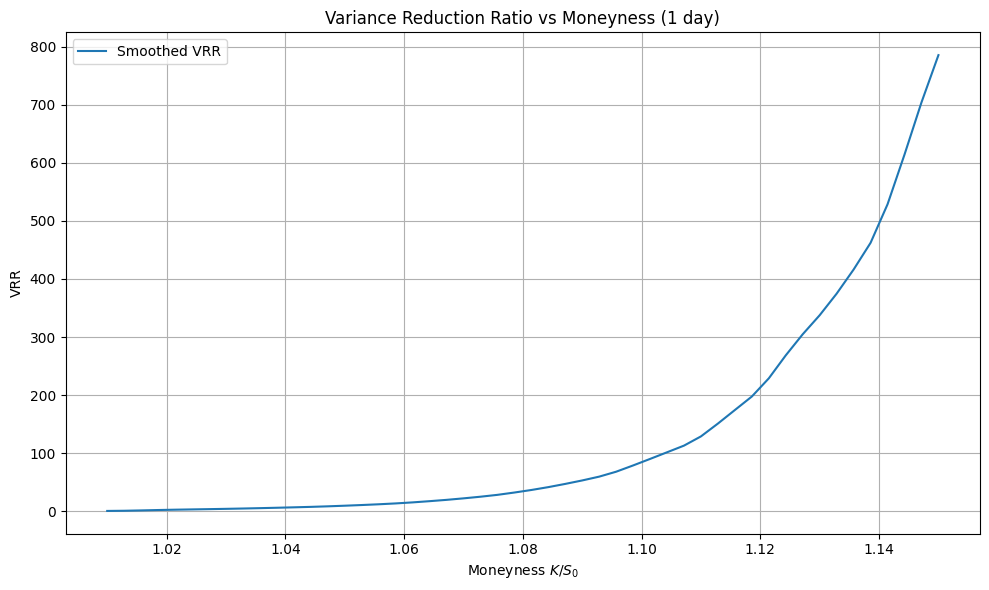} 
    \caption{VRR vs. Moneyness for $T=1/252$ (1 Day). The VRR grows exponentially with moneyness, exceeding 2500 for deep OTM strikes.}
    \label{fig:vrr_short}
\end{figure}

\begin{figure}[H]
    \centering
    \includegraphics[width=0.8\textwidth]{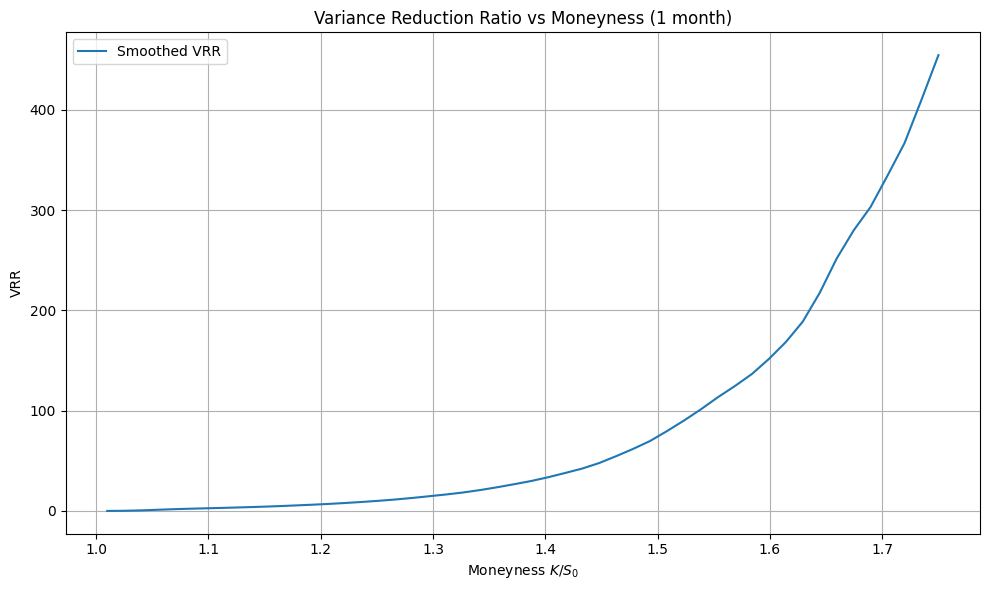}
    \caption{VRR vs. Moneyness for $T=21/252$ (1 Month). Significant variance reduction is maintained, peaking around 450.}
    \label{fig:vrr_medium}
\end{figure}

\begin{figure}[H]
    \centering
    \includegraphics[width=0.8\textwidth]{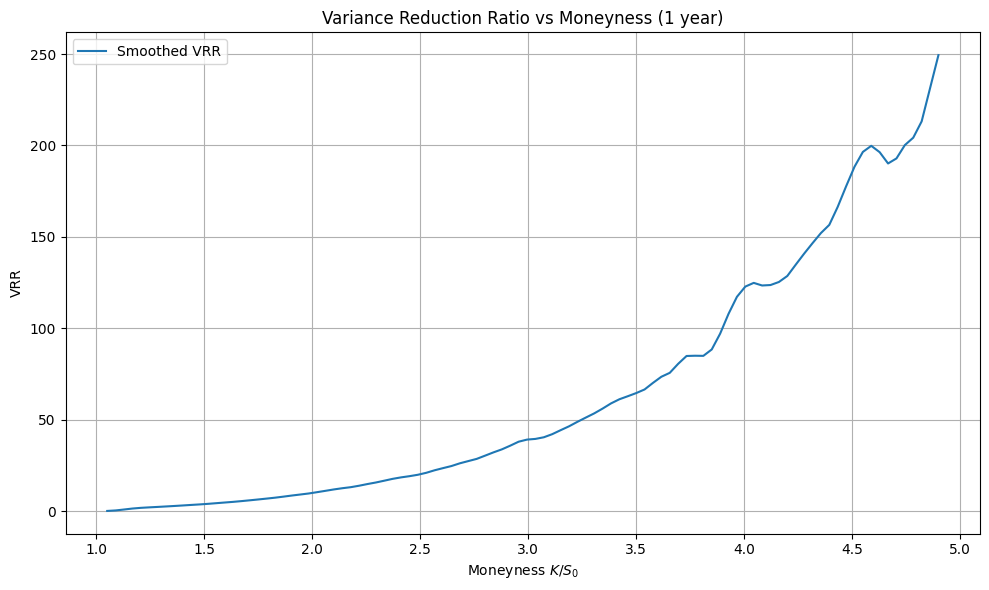}
    \caption{VRR vs. Moneyness for $T=1$ (1 Year). The VRR saturates below 250, indicating the diminishing effectiveness of the specific drift as the time horizon allows for more complex path dynamics.}
    \label{fig:vrr_long}
\end{figure}

\paragraph{Discussion.}
Figures \ref{fig:vrr_short} through \ref{fig:vrr_long} illustrate the exponential efficiency of the proposed scheme.
\begin{itemize}
    \item In the \textbf{Short Maturity} case (Fig. \ref{fig:vrr_short}), the VRR explodes as $K$ increases, reaching values over 2500. This confirms that our IS density $Q(\bar{h})$ captures the large deviation optimal path almost perfectly in this regime.
    \item In the \textbf{Medium Maturity} case (Fig. \ref{fig:vrr_medium}), the method remains highly effective (VRR $\sim 450$), demonstrating robustness.
    \item In the \textbf{Long Maturity} case (Fig. \ref{fig:vrr_long}), while still providing useful variance reduction (VRR $\sim 100-250$), the growth rate slows down. This is physically intuitive: over long horizons, the variance process can fluctuate significantly away from its initial value, and a drift proportional to $\sqrt{V_t}$ (based on the slowly-varying assumption) captures the dominant behavior but may miss second-order path fluctuations. Nevertheless, the method remains superior to standard MC.
\end{itemize}

\section{Conclusion} \label{sec:conclusion}

In this paper, we have developed and analyzed an asymptotically optimal importance sampling strategy for pricing European call options under the Heston model. By leveraging the Large Deviation Principle, we constructed a state-dependent change of measure where the drift of the driving Brownian motions is scaled by the instantaneous volatility.

Our theoretical contributions are twofold. First, in the \textbf{short-maturity regime}, we bridged the gap between the known implied volatility asymptotics and Monte Carlo variance reduction, proving via Riccati analysis that our scheme achieves logarithmic efficiency. Second, and more significantly, we introduced a novel \textbf{slow mean-reversion scaling} ($\delta = \varepsilon^{-2}$) for the \textbf{deep OTM regime}. We demonstrated that under this scaling, the stochastic volatility contributes non-trivially to the rate function, and our specific drift design is required to match the asymptotic decay of the second moment.

Numerical experiments confirmed our theoretical findings, showing substantial variance reduction ratios—spanning from two to three orders of magnitude—particularly in the regimes where standard estimators fail most severely. These results highlight the power of combining large deviation theory with singular perturbation techniques to design efficient simulation algorithms for complex path-dependent derivatives. Future work may extend this framework to path-dependent options (e.g., Asian or Barrier options) or Rough Volatility models, where the non-Markovian nature of volatility poses new challenges for importance sampling design.

\bibliographystyle{plainnat} 
\bibliography{references}

\section*{Appendix}
\addcontentsline{toc}{section}{Appendix}
\setcounter{subsection}{0}
\renewcommand{\thesubsection}{\Alph{subsection}}

\subsection{Derivation of the Short-Maturity Term I} \label{app:sm_term_I}

In this appendix, we derive the asymptotic limit of the integrated variance moment presented in Equation \eqref{eq:limit_term_I}. We aim to evaluate the limit of the moment generating function at maturity $T$, defined as:
\
\begin{equation}
    L_I \coloneqq \lim_{T \to 0}  T \log{\mathbb{E}^{\widetilde{\mathbb{P}}} \left[ \exp\left( q C(\bar{h}) \int_0^T V_s ds \right) \mid V_0 = v_0 \right]}
\end{equation}
\paragraph{Feynman-Kac and Riccati Equation}
Let $F_I (t,v)\coloneqq \mathbb{E}^{\widetilde{\mathbb{P}}} \left[ \exp\left( q C(\bar{h}) \int_0^t V_s ds \right) \mid V_0 = v \right]$ denote the expectation term over the horizon $t$. By the Feynman-Kac theorem, $F_I(t, v)$ satisfies the following partial differential equation:
\begin{equation} \label{eq:fk_pde_sm_app}
    \frac{\partial F_I}{\partial t} = \frac{1}{2} \sigma^2 v \frac{\partial^2 F_I}{\partial v^2} + (\kappa\theta - \tilde{\kappa}_I v) \frac{\partial F_I}{\partial v} + q C(\bar{h}) v F_I,
\end{equation}
subject to the initial condition $F_I(0, v) = 1$. The effective mean-reversion is $\tilde{\kappa}_I=\kappa - 2 \rho \sigma$.
Given the affine structure of the Heston model, we adopt the exponential-affine ansatz $F_I(t, v) = \exp(\phi_I(t) + v \psi_I(t))$. Substituting this ansatz into \eqref{eq:fk_pde_sm_app} yields a system of ODEs, where the Riccati equation for $\psi_I(t)$ is decoupled:
\begin{equation} \label{eq:app_riccati}
    \psi_I'(t) = C_0 - C_1 \psi_I(t) + C_2 \psi_I(t)^2, \quad \psi_I(0) = 0,
\end{equation}
Here, the coefficients are defined as:
\begin{equation}
C_{I,0} =  q \left( 1 + 2\bar{h} + \frac{\bar{h}^2}{\bar{\rho}^2} \right), \quad
   C_{I,1}  = \tilde{\kappa}_I, \quad
    C_{I,2} = \frac{\sigma^2}{2}.
\end{equation}
To linearize \eqref{eq:app_riccati}, we employ the standard variable transformation $\psi_I(t) = -\frac{2}{\sigma^2} \frac{u_I'(t)}{u_I(t)}$.
The function $u_I(t)$ then satisfies the second-order linear ODE:
 \begin{equation} \label{eq:app_ode_u}
    u_I''(t) + C_{I,1} u_I'(t) + C_{I,0} C_{I,2} u_I(t) = 0,
\end{equation}
with initial conditions $u_I(0) = 1$ and $u_I'(0) = 0$ (implied by $\psi_I(0)=0$).

\paragraph{Asymptotic Analysis}
We now consider the regime where the maturity $T \to 0$. Recall the importance sampling drift is chosen as $\bar{h} = -\frac{k}{\theta T}$. This implies the coefficient $C_{I,0}$ scales as $O(T^{-2})$:
\begin{equation}
    C_{I,0}  \sim \frac{q k^2}{\bar{\rho}^2 \theta^2 T^2}.
\end{equation}
Consequently, the discriminant of the characteristic equation is dominated by the product term $4 C_{I,0} C_{I,2}$, leading to a large negative value:
\begin{equation} \label{eq:app_discriminant}
    \Delta_I = C_{I,1}^2 - 4 C_{I,0} C_{I,2} \sim - \frac{2 \sigma^2 q k^2}{\bar{\rho}^2 \theta^2 T^2}.
\end{equation}
The negative discriminant indicated that the system operates in a highly oscillatory regime. Let $\omega_I \coloneqq \frac{1}{2}\sqrt{-\Delta_I}$. The asymptotic frequency is given by:
\begin{equation} \label{eq:app_omega}
    \omega_I \sim  \frac{\sigma k \sqrt{2q}}{2 \bar{\rho} \theta T}.
\end{equation}
The general solution for  \eqref{eq:app_ode_u} is thus:
\begin{equation}
    u_I(t) = e^{C_{I,1} t / 2} \left( \cos(\omega_I t) + \frac{C_{I,1}}{2\omega_I} \sin(\omega_I t) \right).
\end{equation}
In the limit $t = T \rightarrow 0$, the damping term $e^{C_{I,1} t / 2} \rightarrow 1$, and the ratio $\frac{C_{I,1}}{2\omega_I} \sin(\omega_I t) \rightarrow 0$. We approximate the solution and its derivative as:
\begin{equation}
    u_I(T) \sim \cos{\omega_I T}, \quad u_I'(T) \sim -\omega_I \sin{\omega_I T}.
\end{equation}
Substituting these back into the transformation for $\psi_I(T)$:
\begin{equation}
    \psi_I(T) = -\frac{2}{\sigma^2} \frac{u_I'(T)}{u_I(T)}\sim \frac{2\omega_I}{\sigma^2} \tan(\omega_I T).
\end{equation}

Finally, we compute the limit of the exponent. Note that the constant term  $\phi_I(T)$ scales as $O(1)$ and vanishes under the $T \log(\cdot)$ scaling. The limit is determined entirely by the $v_0 \psi_I(T)$ term:
\begin{equation}
    L_I = v_0 \lim_{T \to 0} T \psi_I(T) =  v_0 \cdot \frac{k \sqrt{2q}}{\sigma \bar{\rho} \theta} \tan\left( \frac{\sigma k \sqrt{2q}}{2 \bar{\rho} \theta} \right).
\end{equation}
Multiplying  by the factor $1/q$ from the H\"older decomposition, we obtain the final result:
\begin{equation}
\lim_{T \rightarrow 0} (\text{Term I}) = \frac{v_0 k \sqrt{2/q}}{\sigma \theta \bar{\rho}} \tan\left( \frac{\sigma k \sqrt{2q}}{2 \theta \bar{\rho}} \right).
\end{equation}

\subsection{Derivation of the Short-Maturity $\Gamma_{II}(p)$} \label{app:sm_term_II}

In this section, we determine the limiting SCGF for the log-price process under the auxiliary measure $\widetilde{\mathbb{P}}$. Our objective is to evaluate the limit:
\begin{equation}
    \Gamma_{II}(p) \coloneqq \lim_{T \to 0} T \log \mathbb{E}^{\widetilde{\mathbb{P}}} \left[ \exp\left( \frac{p}{T} (X_T - X_0) \right) \right].
\end{equation}

\paragraph{Feynman-Kac and Riccati Equation}
Let $F_{II}(t, v) = \mathbb{E}^{\widetilde{\mathbb{P}}}[\exp(\frac{p}{T}(X_T - X_t)) \mid V_t = v]$ denote the moment generating function. Under $\widetilde{\mathbb{P}}$, the drift of $X_t$ is modified to $(\frac{3}{2} + \bar{h})V_t$. By the Feynman-Kac theorem, $F_{II}$ satisfies the following PDE:
\begin{equation} \label{eq:fk_pde_term2}
    \frac{\partial F_{II}}{\partial t} = \frac{1}{2}\sigma^2 v \frac{\partial^2 F_{II}}{\partial v^2} + \left( \kappa\theta - \tilde{\kappa}_{II} v \right) \frac{\partial F_{II}}{\partial v} +  \left[\left( \frac{3}{2} + \bar{h} \right) \frac{p}{T} + \frac{1}{2} \bar{\rho}^2 \left( \frac{p}{T} \right)^2\right] v  F_{II},
\end{equation}
subject to $F_{II}(0,v)=1$. The effective mean-reversion speed is $\tilde{\kappa}_{II} = \kappa - 2\rho\sigma$.
Applying the affine ansatz $F_{II}(t, v) = \exp(\phi_{II}(t) + v\psi_{II}(t))$ yields the Riccati ODE for $\psi_{II}(t)$:
\begin{equation}
    \psi_{II}'(t) = C_{II,0}(T) - C_{II,1}(T) \psi_{II}(t) + C_{II,2} \psi_{II}(t)^2, \quad \psi_{II}(0) = 0.
\end{equation}
Unlike the variance moment case (Appendix \ref{app:sm_term_I}), the coefficients here depend explicitly on $T$:
\begin{equation}
    C_{II,0}(T) = \frac{p}{T} \left( \frac{3}{2} + \bar{h} \right) + \frac{1}{2} \left(\frac{p}{T}\right)^2 \bar{\rho}^2, \quad
    C_{II,1}(T) = \tilde{\kappa}_{II} - \frac{p}{T} \rho \sigma, \quad
    C_{II,2} = \frac{\sigma^2}{2}.
\end{equation}
Using the transformation $\psi_{II}(t)=-\frac{2}{\sigma^2} \frac{u_{II}'(t)}{u_{II}(t)}$, we obtain the second-order linear ODE:
\begin{equation}
    u_{II}''(t)+C_{II,1}(T)u_{II}'(t)+C_{II,0}(T)C_{II,2}u_{II}(t)=0
\end{equation}
with initial conditions $u_{II}(0)=1$ and $u_{II}'(0)=0$. 

\paragraph{Asymptotic Analysis}
We now analyze the coefficients in the limit $T \rightarrow 0$. Substituting the importance sampling drift $\bar{h}$, we observe the following scaling behaviors:
\begin{equation}
    C_{II,0}(T) \approx \frac{1}{T^2} \left( -\frac{pk}{\theta} + \frac{p^2 \bar{\rho}^2}{2} \right), \quad C_{II,1}(T) \approx -\frac{p \rho \sigma}{T}.
\end{equation}
A distinct feature of this regime is that the linear coefficient  $C_{II,1}(T)$ scales as $O(T^{-1})$. Consequently, its square contributes to the leading order of the discriminant:
\begin{equation}
    \Delta_{II} = C_{II,1}(T)^2 - 4C_{II,0}(T)C_{II,2} \approx \frac{1}{T^2}\left[ \left( -{p\rho\sigma} \right)^2 - 2 \sigma^2 \left( -\frac{pk}{\theta} + \frac{p^2 \bar{\rho}^2}{2} \right) \right].
\end{equation}
We define the scaled discriminant $\hat{\Delta}_{II} \coloneqq T^2\Delta_{II} = \sigma^2 \left( p^2 (2\rho^2 - 1) + \frac{2pk}{\theta} \right)$. The nature of the solution depends on the sign of $\hat{\Delta}_{II}(p)$:
\begin{itemize}
    \item \textbf{Exponential Regime ($\hat{\Delta}_{II} > 0$)}: The characteristic roots are real. The solution involves hyperbolic functions, and the limit yields:
    \begin{equation}
        \Gamma_{II}(p)=\frac{v_0}{\sigma^2}\left(-p\rho\sigma + \sqrt{\hat{\Delta}_{II}}\tanh{\left(\frac{\sqrt{\hat{\Delta}_{II}}}{2} \right)} \right).
    \end{equation}
    \item \textbf{Linear Regime ($\hat{\Delta}_{II} = 0$)}: The ODE degenerates, and the solution $u_{II}(T)$ behaves such that the log-derivative is linear. This occurs at the critical points $p \in \{0,p^*_{II}\}$, yielding:
    \begin{equation}
        \Gamma_{II}(p)=-\frac{v_0p\rho}{\sigma}.
    \end{equation}
    \item \textbf{Oscillatory Regime ($\hat{\Delta}_{II} < 0$)}: The characteristic roots are imaginary. The solution involves trigonometric functions similar to Appendix \ref{app:sm_term_I}. The asymptotic limit is:
\begin{equation}
    \Gamma_{II}(p) = \frac{v_0}{\sigma^2} \left( -p \rho \sigma + \sqrt{-\hat{\Delta}_{II}} \tan\left( \frac{\sqrt{-\hat{\Delta}_{II}}}{2} \right) \right).
\end{equation}
\end{itemize}
These three cases constitute the piecewise definition of the SCGF presented in Equation \ref{eq:gamma3_explicit_unified}

\subsection{Derivation of the Deep OTM SCGF $\Gamma_1^\varepsilon(p)$} \label{app:dotm_riccati}

In this appendix, we rigorously derive the limiting SCGF for the Deep OTM regime, denoted as $\Gamma_1^\varepsilon(p)$, and determine its domain $(p^\varepsilon_-, p^\varepsilon_+)$. Recall the definition:
\begin{equation}
    \Gamma_1^\varepsilon(p) \coloneqq \lim_{\varepsilon \to 0} \varepsilon^2 \log \mathbb{E}^{\mathbb{P}^\varepsilon} \left[ \exp\left( \frac{p}{\varepsilon^2} (X^\varepsilon_T - X^\varepsilon_0) \right) \right].
\end{equation}
The dynamics under the scaling $\delta = \varepsilon^{-2}$ follow System \eqref{eq:deep_otm_system}.

\paragraph{Measure Change}
To facilitate the expectation calculation, we perform a Girsanov transformation to eliminate the stochastic integral with respect to  $W^1$. Let $\widetilde{\mathbb{Q}}$ be the measure defined by the Radon-Nikodym derivative:
\begin{equation}
    \frac{d\widetilde{\mathbb{Q}}}{d\mathbb{P}^\varepsilon} \coloneqq \exp\left( \frac{p\rho}{\varepsilon^{1.5}} \int_0^T \sqrt{V^\varepsilon_t} dW^1_t - \frac{1}{2} \frac{p^2\rho^2}{\varepsilon^3} \int_0^T V^\varepsilon_t dt \right).
\end{equation}
Under $\widetilde{\mathbb{Q}}$, the expectation reduces to a functional of the the integrated variance with an effective coefficient $J^\varepsilon$. This coefficient collects contributions from the original drift, the measure change compensator, and the quadratic variation from $W^2$:
\begin{equation}
    J^\varepsilon = -\frac{p}{2\varepsilon^2} + \frac{1}{2}\left(\frac{p\rho}{\varepsilon^{1.5}}\right)^2 + \frac{1}{2}\left(\frac{p\bar{\rho}}{\varepsilon^{1.5}}\right)^2 = \frac{p^2}{2\varepsilon^3} - \frac{p}{2\varepsilon^2}.
\end{equation}
\paragraph{Feynman-Kac and Riccati Equation}
Let $F^{\varepsilon}(t, v) = \mathbb{E}^{\widetilde{\mathbb{Q}}} [ \exp( \int_0^t J^{\varepsilon} V^\varepsilon_s ds ) \mid V^\varepsilon_0 = v ]$. By the Feynman-Kac theorem, $F^{\varepsilon}(t, v)$ satisfies the PDE:
\begin{equation} \label{eq:dotm_fk_pde}
    \frac{\partial F^{\varepsilon}}{\partial t} = \frac{1}{2} (\sigma \varepsilon^{1.5})^2 v \frac{\partial^2 F^{\varepsilon}}{\partial v^2} + (\kappa \varepsilon^3 \theta - \tilde{\kappa}^\varepsilon v) \frac{\partial F^{\varepsilon}}{\partial v} + J_\varepsilon v F^{\varepsilon},
\end{equation}
subject to $F^{\varepsilon}(0, v) = 1$. Here, $\tilde{\kappa}^\varepsilon = \kappa \varepsilon^2 - p \rho \sigma$ is the effective mean-reversion speed.
Substituting the affine ansatz $F^{\varepsilon}(t, v) = \exp(\phi^\varepsilon(t) + v \psi^\varepsilon(t))$ into \eqref{eq:dotm_fk_pde}, we obtain the Riccati ODE for $\psi^\varepsilon(t)$:
\begin{equation} \label{eq:dotm_riccati_psi}
    (\psi^\varepsilon)'(t) = J^\varepsilon - \tilde{\kappa}^\varepsilon \psi^\varepsilon(t) + \frac{1}{2} (\sigma \varepsilon^{1.5})^2 \psi^\varepsilon(t)^2, \quad \psi^\varepsilon(0)=0.
\end{equation}
To solve this, we use the transformation $\psi^\varepsilon(t) = -\frac{2}{\sigma^2 \varepsilon^3} \frac{(u^\varepsilon)'(t)}{u^\varepsilon(t)}$, we obtain linear ODE:
\begin{equation}
    (u^\varepsilon)''(t) + \tilde{\kappa}^\varepsilon (u^\varepsilon)'(t) + \left( J^\varepsilon \frac{\sigma^2 \varepsilon^3}{2} \right) u^\varepsilon(t) = 0.
\end{equation}
with initial conditions $u^\varepsilon(0)=1, (u^\varepsilon)'(0)=0$.

\paragraph{Asymptotic Analysis}
We now consider the limit $\varepsilon \to 0$. The constant term in the ODE simplifies as the $\varepsilon^3$ factors cancel:
\begin{equation}
    \lim_{\varepsilon \rightarrow 0} \left(J^\varepsilon \frac{\sigma^2 \varepsilon^3}{2}  \right) = \lim_{\varepsilon \rightarrow 0} \left(\frac{p^2}{2\varepsilon^3} - \frac{p}{2\varepsilon^2}\right) \frac{\sigma^2 \varepsilon^3}{2} = \frac{p^2 \sigma^2}{4}.
\end{equation}
Similarly, the damping term $\tilde{\kappa}^\varepsilon \rightarrow -p \rho \sigma$. The limiting ODE becomes:
\begin{equation}
    (u^\varepsilon)''(t) - p \rho \sigma (u^\varepsilon)'(t) + \frac{p^2 \sigma^2}{4} u^\varepsilon(t) = 0.
\end{equation}

The discriminant of the characteristic equation is $\Delta^\varepsilon = (-p\rho\sigma)^2 - p^2 \sigma^2 = -p^2 \sigma^2 \bar{\rho}^2$. Since $\Delta^\varepsilon < 0$ for $p \neq 0$, the solution is strictly oscillatory. Define the frequency $\omega^\varepsilon = \frac{|p| \bar{\rho} \sigma}{2}$. The solution evaluated at maturity $T$ is:
\begin{equation}
    u^\varepsilon(T) \approx e^{\frac{p \rho \sigma T}{2}} \cos(\omega^\varepsilon T),
\end{equation}
Substituting this back into expression for $\psi^\varepsilon(T)$. Recall that the initial variance scales as $V^\varepsilon_0 = \varepsilon v_0$. Therefore, the limit of the SCGF is:
 \begin{equation}
     \Gamma_1^\varepsilon(p) = \lim_{\varepsilon \rightarrow 0} \varepsilon^2 (\varepsilon v_0) \psi^\varepsilon(T)=\frac{v_0 p}{ -\rho \sigma + \bar{\rho} \sigma \cot\left( \frac{p \bar{\rho} \sigma T}{2} \right) }.
 \end{equation}
 
\paragraph{Effective Domain}
The function $\Gamma_1^\varepsilon(p)$ is well-defined in the interval containing zero where the denominator does not vanish. The boundaries $(p^\varepsilon_-, p^\varepsilon_+)$ are determined by the first singularities $(\xi^\varepsilon_-, \xi^\varepsilon_+)$:
\begin{itemize}
    \item \textbf{Case $\rho > 0$:} $\xi^\varepsilon_+ = \arctan(\bar{\rho}/\rho), \space\xi^\varepsilon_- = \arctan(\bar{\rho}/\rho) - \pi$. Thus,
    \begin{equation}
            p^\varepsilon_+ = \frac{2}{\sigma \bar{\rho} T} \arctan\left(\frac{\bar{\rho}}{\rho}\right), \quad p^\varepsilon_- = \frac{2}{\sigma \bar{\rho} T} \left( \arctan\left(\frac{\bar{\rho}}{\rho}\right) - \pi \right).
        \end{equation}
    \item \textbf{Case $\rho < 0$:} $\xi^\varepsilon_+ = \arctan(\bar{\rho}/\rho), \space\xi^\varepsilon_- = \arctan(\bar{\rho}/\rho) - \pi$. Thus, 
    \begin{equation}
            p^\varepsilon_+ = \frac{2}{\sigma \bar{\rho} T} \arctan\left(\frac{\bar{\rho}}{\rho}\right), \quad p^\varepsilon_- = \frac{2}{\sigma \bar{\rho} T} \left( \arctan\left(\frac{\bar{\rho}}{\rho}\right) - \pi \right).
        \end{equation}
    \item \textbf{Case $\rho = 0$:} The condition becomes $\cot(\xi^\varepsilon_{\pm}) = 0$, yielding $p^\varepsilon_\pm = \pm \frac{\pi}{\sigma T}$.
\end{itemize}
Within $(p^\varepsilon_-, p^\varepsilon_+)$, $\Gamma_1^\varepsilon(p)$ is essentially smooth, satisfying the requirements of the Gärtner-Ellis theorem.

\subsection{Derivation of the Deep OTM Term I} \label{app:dotm_term_I}

In this appendix, we derive the asymptotic limit of the Deep OTM Term I. We aim to compute the limit:
\begin{equation}
    \lim_{\varepsilon \to 0} \frac{\varepsilon^2}{q} \log \mathbb{E}^{\widetilde{\mathbb{P}}^\varepsilon} \left[ \exp\left( q C(\bar{h}) \int_0^T V^\varepsilon_t \frac{dt}{\varepsilon} \right) \right].
\end{equation}

\paragraph{Feynman-Kac and Riccati Equation}
Let $F^\varepsilon_I(t, v) = \mathbb{E}^{\widetilde{\mathbb{P}}^\varepsilon} [ \exp( q C(\bar{h}) \varepsilon^{-1} \int_0^t V^\varepsilon_s ds ) \mid V^\varepsilon_0 = v ]$. By the Feynman-Kac theorem, $F^\varepsilon_I(t, v)$ satisfies the following partial differential equation (PDE):
\begin{equation} \label{eq:dotm_fk_pde_I}
    \frac{\partial F^\varepsilon_I}{\partial t} = \frac{1}{2} (\sigma \varepsilon^{1.5})^2 v \frac{\partial^2 F^\varepsilon_I}{\partial v^2} + (\kappa \varepsilon^3 \theta - \tilde{\kappa}^\varepsilon_{I} v) \frac{\partial F^\varepsilon_I}{\partial v} + \frac{q C(\bar{h})}{\varepsilon} v F^\varepsilon_I,
\end{equation}
subject to the initial condition $F^\varepsilon_I(0, v) = 1$. Here, the effective mean-reversion speed under $\widetilde{\mathbb{P}}^\varepsilon$ is $\tilde{\kappa}^\varepsilon_{I} = \kappa \varepsilon^2 - 2\rho \sigma \varepsilon$.

By the affine structure, the solution takes the form $F^\varepsilon_I (t, v) = \exp(\phi^\varepsilon_I (t) + v \psi^\varepsilon_I (t))$. Substituting this ansatz into \eqref{eq:dotm_fk_pde_I} leads to the Riccati ODE for $\psi^\varepsilon_I(t)$:
\begin{equation} \label{eq:dotm_riccati_ode}
    \left({\psi^\varepsilon_I}\right)'(t) = C_{I,0}^\varepsilon - C_{I,1}^\varepsilon \psi^\varepsilon_I (t) + C_{I,2}^\varepsilon \psi^\varepsilon_I(t)^2, \quad \psi^\varepsilon_I(0) = 0.
\end{equation}
The coefficients correspond to the terms in the PDE:
\begin{equation}
    C_{I,0}^\varepsilon = \frac{q C(\bar{h})}{\varepsilon}, \quad
    C_{I,1}^\varepsilon = \tilde{\kappa}^\varepsilon_{I}, \quad
   C_{I,2}^\varepsilon = \frac{1}{2} (\sigma \varepsilon^{1.5})^2 .
\end{equation}

We linearize the equation using the transformation $\psi^\varepsilon_I(t) = -\frac{1}{C_{I,2}^\varepsilon} \frac{(u^\varepsilon_I)'(t)}{u^\varepsilon_I(t)}$. Substituting this into \eqref{eq:dotm_riccati_ode} yields the second-order linear ODE:
\begin{equation}
    (u^\varepsilon_I )''(t) + C_{I,1}^\varepsilon (u^\varepsilon_I )'(t) + C_{I,0}^\varepsilon C_{I,2}^\varepsilon u^\varepsilon_I (t) = 0,
\end{equation}
with initial conditions $u^\varepsilon_I (0)=1, (u^\varepsilon_I )'(0)=0$.

\paragraph{Asymptotic Analysis}
We now analyze the coefficients in the small-noise limit. The drift parameter is $\bar h = -1/(\varepsilon \theta T)$. Substituting this into the constant term, we find:
\begin{equation}
    C_{I,0}^\varepsilon C_{I,2}^\varepsilon \approx \left( \frac{q}{\varepsilon^3 \theta^2 T^2 \bar{\rho}^2} \right) \left( \frac{\sigma^2 \varepsilon^3}{2} \right) = \frac{q \sigma^2}{2 \theta^2 T^2 \bar{\rho}^2}.
\end{equation}
Crucially, the $\varepsilon^3$ scaling factors cancel exactly, leaving and $O(1)$ constant term. The discriminant of the characteristic equation becomes:
\begin{equation}
    \Delta^\varepsilon_I = (C_{I,1}^\varepsilon)^2 - 4 C_{I,0}^\varepsilon C_{I,2}^\varepsilon \approx (\kappa \varepsilon^2 - 2\rho \sigma \varepsilon)^2 -  \frac{2 q \sigma^2}{\theta^2 T^2 \bar{\rho}^2}.
\end{equation}

As $\varepsilon \to 0$, the first term $(C_{I,1}^\varepsilon)^2 \sim O(\varepsilon^2)$ vanishes, and the discriminant converges to a negative constant:
\begin{equation}
    \lim_{\varepsilon \rightarrow 0}\Delta^\varepsilon_I = \hat{\Delta}_I^\varepsilon  \coloneqq -\frac{2 q \sigma^2}{\theta^2 T^2 \bar{\rho}^2}.
\end{equation}

This negative discriminant implies an oscillatory solution regime. Let $\omega_I^\varepsilon \coloneqq \frac{1}{2}\sqrt{-\hat{\Delta}_I^\varepsilon}$. In the limit, the damping term vanishes ($C_{I,1}^\varepsilon \to 0$), and the solution behaves as:
\begin{equation}
   u^\varepsilon_I (T) \sim \cos(\omega_I^\varepsilon T)
\end{equation}
The logarithm derivative is thus:
\begin{equation}
    \psi^\varepsilon_I (T) = -\frac{2}{\sigma^2 \varepsilon^3} \frac{(u^\varepsilon_I )'(T)}{u^\varepsilon_I (T)} \sim \frac{2\omega_I^\varepsilon}{\sigma^2 \varepsilon^3} \tan(\omega_I^\varepsilon T).
\end{equation}
Finally, we compute the limit of the log-expectation $L^\varepsilon_I$:
\begin{equation}
    \lim_{\varepsilon \to 0} \frac{\varepsilon^2 (\varepsilon v_0)}{q}\psi^\varepsilon_I(T) =
    \frac{1}{q} \cdot \frac{v_0 \sqrt{2q}}{\sigma \theta \bar{\rho} T} \tan\left( \frac{\sigma \sqrt{2q}}{2 \theta \bar{\rho}} \right).
\end{equation}
\subsection{Derivation of the Deep OTM Auxiliary SCGF $\Gamma_3^\varepsilon(p)$} \label{app:dotm_term_II}

In this appendix, we derive the limiting SCGF under the auxiliary measure $\widetilde{\mathbb{P}}^\varepsilon$, denoted as $\Gamma^\varepsilon_{II} (p)$. This function is required to evaluate Term II in Section \ref{sec:dotm_second_moment}. We define:
\begin{equation}
    \Gamma^\varepsilon_{II}(p) \coloneqq \lim_{\varepsilon \to 0} \varepsilon^2 \log \mathbb{E}^{\widetilde{\mathbb{P}}^\varepsilon} \left[ \exp\left( \frac{p}{\varepsilon^2} (X^\varepsilon_T - X^\varepsilon_0) \right) \right].
\end{equation}

\paragraph{Feynman-Kac and Riccati Equation}
The expectation can be reduced to a functional of the integrated variance. We define the effective coefficient $J^\varepsilon_{II}$ which collects the contributions from the drift of the scaled log-price and the quadratic variation compensator.
Under $\widetilde{\mathbb{P}}^\varepsilon$, the drift of $X^\varepsilon_t$ is dominated by $\varepsilon \bar{h} V_t \approx -\frac{1}{\varepsilon \theta T} V^\varepsilon_t$. Combined with the exponent $p/\varepsilon^2$ and the standard quadratic variation $\frac{p^2}{2\varepsilon^3}$, we define:
\begin{equation}
    J^\varepsilon_{II} \coloneqq \frac{p^2}{2\varepsilon^3} - \frac{p}{\varepsilon^3 \theta T} = \frac{1}{\varepsilon^3} \left( \frac{p^2}{2} - \frac{p}{\theta T} \right).
\end{equation}
Both terms are of order $O(\varepsilon^{-3})$, which balances the diffusion coefficient in the limit.

Let $F^\varepsilon_{II}(t, v) = \mathbb{E}^{\widetilde{\mathbb{P}}^\varepsilon} [ \exp( \int_0^t J^\varepsilon_{II} V^\varepsilon_s ds ) \mid V^\varepsilon_0 = v ]$. By the Feynman-Kac theorem, $F^\varepsilon_{II}$ satisfies the following PDE:
\begin{equation} \label{eq:dotm_fk_pde_II}
    \frac{\partial F^\varepsilon_{II}}{\partial t} = \frac{1}{2} ( \sigma \varepsilon^{1.5})^2 v \frac{\partial^2 F^\varepsilon_{II}}{\partial v^2} + (\kappa \varepsilon^3 \theta - \tilde{\kappa}^\varepsilon_{II} v) \frac{\partial F^\varepsilon_{II}}{\partial v} + J^\varepsilon_{II} v F^\varepsilon_{II},
\end{equation}
subject to $F^\varepsilon_{II}(0, v) = 1$. Here,  $\tilde{\kappa}^\varepsilon_{II} = \kappa \varepsilon^2 - p \rho \sigma$ represents the effective mean-reversion speed under the measure $\widetilde{\mathbb{P}}^\varepsilon$.

Substituting the affine ansatz $F^\varepsilon_{II}(t, v) = \exp(\phi^\varepsilon_{II}(t) + v \psi^\varepsilon_{II}(t))$ into the PDE yields the Riccati ODE for $\psi^\varepsilon_{II}(t)$:
\begin{equation} \label{eq:dotm_riccati_ode_3}
    \left({\psi^\varepsilon_{II}}\right)'(t) = J^\varepsilon_{II} - \tilde{\kappa}^\varepsilon_{II} \psi^\varepsilon_{II}(t) + \frac{1}{2} ( \sigma \varepsilon^{1.5})^2 \psi^\varepsilon_{II}(t)^2, \quad \psi^\varepsilon_{II}(0) = 0.
\end{equation}

To solve this, we employ the transformation $\psi^\varepsilon_{II}(t) = -\frac{2}{\sigma^2 \varepsilon^3} \frac{(u^\varepsilon_{II})'(t)}{u^\varepsilon_{II}(t)}$. Substituting this into \eqref{eq:dotm_riccati_ode_3} yields the second-order linear ODE for $u_{II}(t)$:
\begin{equation}
    {(u^\varepsilon_{II})}''(t) + \tilde{\kappa}^\varepsilon_{II} {(u^\varepsilon_{II})}'(t) + \left( J^\varepsilon_{II} \frac{\sigma^2 \varepsilon^3}{2} \right) u^\varepsilon_{II}(t) = 0,
\end{equation}
with initial conditions $u^\varepsilon_{II}(0)=1, (u^\varepsilon_{II})'(0)=0$.
\paragraph{Asymptotic Analysis}
We evaluate the coefficients in the limit. The $\varepsilon^3$ terms in the constant coefficient cancel out:
\begin{equation}
\lim_{\varepsilon \to 0}  \left( J^\varepsilon_{II} \frac{\sigma^2 \varepsilon^3}{2} \right) = \frac{\sigma^2}{4}\left(p^2 - \frac{2p}{\theta T}\right).
\end{equation}
The damping term converges to a constant drift Since $\tilde{\kappa}^\varepsilon_{II} \to -p\rho\sigma$. The limiting characteristic equation has the discriminant:
\begin{equation}
    \hat{\Delta}^\varepsilon_{II} \coloneqq (-p\rho\sigma)^2 - 4 \left[ \frac{\sigma^2}{4} \left( p^2 - \frac{2p}{\theta T} \right) \right] = \sigma^2 \left( -p^2 \bar{\rho}^2 + \frac{2p}{\theta T} \right).
\end{equation}
The roots of $\hat{\Delta}^\varepsilon_{II} = 0$ are $p = 0$ and $p^{\varepsilon,*}_{II}=\frac{2}{\theta T \bar{\rho}^2}$. This quadratic structure leads to three distinct regimes for the SCGF $\Gamma^\varepsilon_{II}(p)$:
\begin{itemize}
    \item \textbf{Exponential Regime ($\hat{\Delta}^\varepsilon_{II} > 0$):} For $p \in (0, p^{\varepsilon,*}_{II})$, the roots are real. The solution involves hyperbolic functions, yielding
\begin{equation}
        \Gamma^\varepsilon_{II}(p) = \frac{v_0}{\sigma^2} \left( -p \rho \sigma + \sqrt{\hat{\Delta}^\varepsilon_{II}} \tanh\left( \frac{\sqrt{\hat{\Delta}^\varepsilon_{II}}}{2} T \right) \right).
    \end{equation}
    \item \textbf{Linear Regime ($\hat{\Delta}^\varepsilon_{II} = 0$):} For $p \in \{0, p^{\varepsilon,*}_{II}\}$,the discriminant vanishes. The solution is linear in time, and the SCGF simplifies to: 
\begin{equation}
    \Gamma^\varepsilon_{II}(p)=-\frac{v_0 p \rho}{\sigma}.
\end{equation}
\item \textbf{Oscillatory Regime ($\hat{\Delta}^\varepsilon_{II} < 0$):} Outside the interval $[0,p^{\varepsilon,*}_{II}]$, the roots are imaginary. Define the frequency $\omega^\varepsilon_{II} = \frac{1}{2}\sqrt{-\hat{\Delta}^\varepsilon_{II}}$. The solution involves trigonometric functions, and the limit is
\begin{equation}
        \Gamma^\varepsilon_{II}(p) = \frac{v_0}{\sigma^2} \left( -p \rho \sigma + \sqrt{-\hat{\Delta}^\varepsilon_{II}} \tan\left( \frac{\sqrt{-\hat{\Delta}^\varepsilon_{II}}}{2} T \right) \right).
    \end{equation}
\end{itemize}

This explicitly characterizes $\Gamma^\varepsilon_{II}(p)$ within its effective domain $(p^\varepsilon_{II,-}, p^\varepsilon_{II,+})$.
\end{document}